\documentclass[preprint,10pt]{elsarticle}

\usepackage{latexsym, amsmath, amsfonts, amssymb, amsthm, mathtools, caption, subcaption, graphicx, epsfig, epstopdf, color, array, longtable, calc, multirow, hhline} 

\usepackage[colorlinks,citecolor=blue,,bookmarksnumbered=true,unicode,pdftoolbar=true,pdfmenubar=true,pdfwindowui=true,bookmarksopen=true,hyperfigures=true]{hyperref}
\hypersetup{%
 pdfcreator={pdfcsLaTeX},
 pdfkeywords={Optimal Design, Barycentric Algorithm, Multiplicative Algorithm, Cost Constraints},
 pdftitle={Approximate D-optimal Experimental Design with Simultaneous Size and Cost Constraints},
 pdfauthor={Radoslav Harman},
}

\def\bM{\mathbf M}           

\def\bT{\mathbf T}           

\def\bw{\mathbf w}           
\def\bd{\mathbf d}           

\def\bbQ{\mathbb Q}          
\def\bq{\mathbf q}           
\def\bf{\mathbf f}           
\def\be{\mathbf e}           
\def\tw{\tilde w}            

\def\bbR{\mathbb R}          
\def\bbN{\mathbb N}          

\def\d{\delta}               
\def\X{\mathcal X}           
\def\1{\mathbf 1}
\def\0{\mathbf 0}

\newtheorem{Lemma}{Lemma}
\newtheorem{Theorem}{Theorem}
\newtheorem{Example}{Example}
\newtheorem{Proposition}{Proposition}

\begin{document}

\begin{frontmatter}

\title{Approximate D-optimal Experimental Design\\with Simultaneous Size and Cost Constraints}
\author{Radoslav Harman\corref{cor1}}
\ead{harman@fmph.uniba.sk}
\author{Eva Benkov\'{a}\corref{cor2}}
\cortext[cor1]{Corresponding author.}

\address{Department of Applied Mathematics and Statistics, Faculty of Mathematics, Physics and Informatics, Comenius University, Mlynsk\'{a} dolina, 84248 Bratislava, Slovakia}

\begin{abstract}
Consider an experiment with a finite set of design points representing permissible trial conditions. Suppose that each trial is associated with a cost that depends on the selected design point. In this paper, we study the problem of constructing an approximate D-optimal experimental design with simultaneous restrictions on the size and on the total cost. For the problem of size-and-cost constrained D-optimality, we formulate an equivalence theorem and rules for the removal of redundant design points. We also propose a simple monotonically convergent ``barycentric'' algorithm that allows us to numerically compute a size-and-cost constrained approximate D-optimal design. 
\end{abstract}


\end{frontmatter}

\section{Introduction}

Consider a statistical experiment consisting of a series of trials. In each trial, the observation depends on a design point $x$ selected from a finite design space $\X$ representing all permissible trial conditions. Without loss of generality, we will assume that $\X:=\{1,2, \dots ,n\}$.  

Usually, resources available for the experiment allow us to perform at most $N$ trials, where $N$ is a number known in advance. Suppose that all relevant properties of an experimental design depend only on the numbers $N_x$ of trials performed in individual design points $x \in \X$. Then, we can represent the experimental design by an $n$-dimensional vector $\bw$ of ``design weights'', with components $w_x:=N_x/N$, $x \in \X$. Using this notation, the restriction on the experimental size can be written in the form 
\begin{equation}
  \sum_{x \in \X} w_x \leq 1. \label{eq:size}
\end{equation}

Suppose also that each trial is associated with a known cost $C_x>0$ depending on the corresponding design point $x \in \X$, and the total cost of the experiment cannot exceed a given limit $B>0$. For each $x \in \X$, let $c_x:=\frac{N}{B}C_x$ be the normalized cost. Then, the total cost constraint can be written in the form
\begin{equation}
  \sum_{x \in \X} c_x w_x \leq 1. \label{eq:cost}
\end{equation}

If the values $c_x$ are the same for all $x \in \X$ then the design has, effectively, only a single constraint. However, \eqref{eq:size} and \eqref{eq:cost} may be both relevant if the costs of trials are unequal, which often occurs in practice. For instance, in an application described in \cite{WSB10}, the design space $\X$ represents time, and the cost of conducting a trial is a non-constant function of the time when the observation is sampled. In \cite{PMFB05}, the design space is the set of all combinations of factor levels, some of which are significantly more expensive than others. 

In some situations, the interpretation of the coefficients $c_x$ may be different from direct financial costs. For example, assume that each trial in $x \in \X$ consumes $c_x$ volume units of a specific substrate, as in \cite{ZZ12}. Then, the restriction on the total available volume of the substrate can be captured by an inequality of the form \eqref{eq:cost}. Yet another example of constraints of the type \eqref{eq:cost} can be found in \cite{DF06}, \cite{DFW08} and \cite{Pronzato10}, where the design space corresponds to treatment doses and the costs represent penalties for doses with low efficacy and high toxicity. See \cite{CF95} for further applications of experimental design under constraints. 

In this paper, we will follow an approximate design theory, that is, we will assume that the weights $w_x$, $x \in \X$, are not restricted to the discrete set $\{0,\frac{1}{N},\frac{2}{N},\ldots,1\}$, but can achieve general real values in the interval $[0,1]$. This ``relaxation'' of weights leads to a convex problem of optimal experimental design (the so-called approximate design problem), which is significantly simpler than its discrete version (the so-called exact design problem). For details, see the monographs \cite{Pazman86}, \cite{FH97}, \cite{Pukelsheim06}, and \cite{ADT07}.

The primary goal of this paper is to propose a method of constructing a $D$-optimal design $\bw^*$ in the set of all approximate designs that satisfy both the size and the cost constraints \eqref{eq:size} and \eqref{eq:cost}, that is,
\begin{equation}\label{eq:Dopt}
\bw^* \in \mathrm{argmax}\left\{\phi(\bw): \bw \geq \0_n, \sum_{x \in \X} w_x \leq 1, \sum_{x \in \X} c_x w_x \leq 1\right\},
\end{equation}
where $\geq$ denotes the componentwise comparison and $\0_n$ is the $n$-dimensional zero vector. In \eqref{eq:Dopt}, the function $\phi: [0,\infty)^n \to [0,\infty)$ is the criterion of $D$-optimality defined by $\phi(\bw):=\det^{1/m}( \bM(\bw))$, where
\begin{equation*}
  \bM(\bw) := \sum_{x \in \X} w_x \bf(x)\bf^{\top}(x)
\end{equation*}
is the standardized information matrix of the size $m \times m$. For simplicity, we will assume regularity in the sense that the vectors  $\bf(1),...,\bf(n)$ span $\bbR^m$, and $\bf(x) \neq \0_m$ for all $x \in \X$.

The vectors $\bf(x)$, $x \in \X$, can represent known regressors of a linear regression model with uncorrelated homoscedastic errors. In this case, the $D$-optimal design minimizes the generalized variance of the best linear unbiased estimator of the model parameter. The vectors $\bf(x)$, $x \in \X$, can also be the gradients of the mean-value function of a non-linear regression model with uncorrelated homoscedastic errors. Then, the solution of \eqref{eq:Dopt} is a size-and-cost constrained locally $D$-optimal design (e.g., Chapter 17 in \cite{ADT07} or Chapter 5 in \cite{PP13}).   

 It is possible to show that the criterion of $D$-optimality is continuous, concave, and homogeneous on $[0,\infty)^n$, see, e.g., Chapter 5 and Section 6.2 in \cite{Pukelsheim06}. In particular, the homogeneity of $\phi$ means that $\phi(\gamma\bw)=\gamma\phi(\bw)$ for any $\bw \geq \0_n$ and any $\gamma \geq 0$. Due to the homogeneity of $\phi$, a statistically natural definition of efficiency of a design $\bw^a$ relative to a design $\bw^b$ with $\phi(\bw^b)>0$ is given by $\mathrm{eff}(\bw^a|\bw^b)=\phi(\bw^a)/\phi(\bw^b)$, cf.\ Section 5.15.\ in \cite{Pukelsheim06}. Moreover, criterion of $D$-optimality is monotonic in the sense $\phi(\bw^a) \leq \phi(\bw^b)$ for any pair $\bw^a$, $\bw^b$ of designs satisfying $\bw^a \leq \bw^b$. 

 Note that for problem \eqref{eq:Dopt} the set of feasible designs is non-empty and compact, therefore the continuity of $\phi$ implies that \eqref{eq:Dopt} has at least one optimal solution $\bw^*$. The assumption $\mathrm{span}\{\bf(1),...,\bf(n)\}=\bbR^m$ entails that $\bM(\bw^*)$ is non-singular, that is, $\phi(\bw^*)>0$. However, for some models the optimal solution of \eqref{eq:Dopt} is not unique.

The assumptions of regularity and properties of $\phi$ imply that
\begin{equation}\label{eq:Doptsize}
 \mathrm{argmax}\left\{\phi(\bw): \bw \geq \0_n, \sum_{x \in \X} w_x \leq 1\right\} = \mathrm{argmax}\left\{\phi(\bw): \bw \geq \0_n, \sum_{x \in \X} w_x = 1\right\}.
\end{equation}
Thus, computing a $D$-optimal design under \eqref{eq:size} is equivalent to computing a standard $D$-optimal design, for which there exist many efficient methods (see \cite{Yu11}, \cite{Sagnol11}, \cite{YBT13}, \cite{LP13}, \cite{PZ14} for some recent results). Similarly, since $c_x>0$ for all $x \in \X$, we have
\begin{equation}\label{eq:Doptcost}
 \mathrm{argmax}\left\{\phi(\bw): \bw \geq \0_n, \sum_{x \in \X} c_x w_x \leq 1\right\}  = \mathrm{argmax}\left\{\phi(\bw): \bw \geq \0_n, \sum_{x \in \X} c_x w_x = 1\right\},
\end{equation}
which is a problem that can be transformed to \eqref{eq:Doptsize} using a suitable change of regressors $\bf(x)$, $x \in \X$; see, e.g., Section 6 in \cite{Elfving52} or the end of Section 10.11 in \cite{ADT07}. However, constructing an approximate optimal design under simultaneous size and cost constraints is more complicated, as we discuss next. 

Let $\bw^s$ be optimal for the size constrained problem \eqref{eq:Doptsize} and let $\bw^c$ be optimal for the cost constrained problem \eqref{eq:Doptcost}. Evidently, if $\bw^s$ satisfies the cost constraint \eqref{eq:cost}, then it is a solution of \eqref{eq:Dopt}. Similarly, if $\bw^c$ satisfies the size constraint \eqref{eq:size}, then it solves \eqref{eq:Dopt}.

Suppose that neither of these two simple cases takes place. Let $\bw^*$ be optimal for \eqref{eq:Dopt}. The homogeneity of $\phi$ and property $\phi(\bw^*)>0$ imply that the two strict inequalities $\sum_x w^*_x < 1$ and $\sum_x c_x w_x^* < 1$ cannot be simultaneously true, that is, $\sum_x w^*_x=1$ or $\sum_x c_x w^*_x=1$.

Assume that $\sum_x w^*_x=1$. Define $\alpha:=\sum_x c_x w^*_x$, $\beta:=\sum_x c_x w^s_x$ and $\gamma:=(\beta-1)/(\beta-\alpha)$. Note that $\beta > 1$, and $\alpha \leq 1$, which means that $\gamma \in (0,1]$.  Let
\begin{equation*}
\bw^{**}:=\gamma\bw^*+(1-\gamma)\bw^s.
\end{equation*}
Clearly, $\sum_x w^{**}_x=1$, since both $\bw^*$ and $\bw^s$ have components summing to one. However, $\sum_x c_x w^{**}_x=1$, i.e., $\bw^{**}$ is feasible for \eqref{eq:Dopt}. At the same time, $\phi(\bw^s) \geq \phi(\bw^*)$. Therefore, since $\bw^{**}$ is a convex combination of $\bw^*$ and $\bw^s$, the concavity of $\phi$ guarantees that $\phi(\bw^{**}) \geq \min\{\phi(\bw^*),\phi(\bw^s)\}=\phi(\bw^*)$. But $\bw^{**}$ is feasible for \eqref{eq:Dopt} and $\bw^*$ is optimal for \eqref{eq:Dopt}. Consequently,  $\bw^{**}$ is also optimal for \eqref{eq:Dopt}.

Using the same reasoning we can prove that if $\sum_x c_x w^*_x=1$, then there also exists a $D$-optimal design $\bw^{**}$ satisfying equalities $\sum_x w^{**}_x=1$ and $\sum_x c_x w^{**}_x=1$. Therefore, it is enough to consider the set $\bbQ^n_+$ of designs $\bw \geq \0_n$ simultaneously satisfying equalities  
\begin{eqnarray}
  \sum_{x \in \X} w_x &=& 1, \label{eq:sizeeq} \\
  \sum_{x \in \X} c_x w_x &=& 1. \label{eq:costeq}
\end{eqnarray}
In other words, once we will be able to find a solution of the ``equality'' size-and-cost constrained problem
 \begin{equation}\label{eq:Dopteq}
\bw^* \in \mathrm{argmax}\left\{\phi(\bw): \bw \geq \0_n, \sum_{x \in \X} w_x = 1, \sum_{x \in \X} c_x w_x = 1\right\},
\end{equation}
we will have an exhaustive method of solving the practically usually more meaningful ``inequality'' size-and-cost constrained problem \eqref{eq:Dopt}. 

If the set $\bbQ^n_+$ of feasible solutions of \eqref{eq:Dopteq} is not empty, it is a convex and compact polyhedron. At the beginning of Section \ref{sec:Theory}, we add some natural assumptions on the normalized costs $c_x$, $x \in \X$, that guarantee $\bbQ^n_+ \neq \emptyset$. Then, it is possible to prove a simple ``equivalence theorem'' for the $D$-optimal size-and-cost constrained design solving \eqref{eq:Dopteq}, as well as some other theoretical properties, cf.\ Section \ref{sec:Theory}.

Analytic solutions of \eqref{eq:Dopteq} are possible only in the simplest cases (such as in Example \ref{ex:simple} at the end of this section). However, there are several general methods of constrained numerical optimization that can be used to develop an efficient algorithm specialized to solve \eqref{eq:Dopteq}.

First, there is a Frank-Wolfe-type ``vertex-direction'' algorithm described in Section 2.2 of \cite{CF95}. This algorithm assumes that at each step, a separate mathematical programming problem is solved. Under \eqref{eq:sizeeq} and \eqref{eq:costeq}, the mathematical programming problem is not difficult, which means that the use of the algorithm would be feasible. However, it can be expected to be even slower than the vertex-direction algorithms for the standard approximate $D$-optimality. 

Another method, which is proposed in \cite{TM01} and generalized in \cite{MTC05}, is motivated by the analytic technique of Lagrange multipliers (cf.\ also \cite{Mikulecka83}). The advantage of this method is that it can be applied to computing designs under a non-linear constraint. For our specific linearly constraint problem, this method is too complicated and rather inefficient, without a proof of convergence.

Next, an interesting possibility is to use an algorithm based on the so-called simplicial decomposition, see \cite{UP07} and references therein. This algorithm is based on alternately solving a linear programming sub-problem and a non-linear restricted master problem which finds the maximum of the objective function over the convex hull of a usually small set of feasible points. In \cite{UP07}, the simplicial decomposition algorithm has been used to compute approximate $D$-optimal designs under box constraints on weights, where, at each step, the master problem is solved by a generalized unconstrained multiplicative algorithm (see \cite{Ucinski05}, cf.\ \cite{HT09}). In a similar way, the simplicial decomposition could be adapted to solving the size-and-cost constrained problem \eqref{eq:Dopteq}.

Approximate $D$-optimal designs under linear constraints can also be computed by modern mathematical programming algorithms, namely maxdet programming (\cite{VBW98}) and semidefinite programming (SDP; cf.\ \cite{BTN87}). These algorithms are very versatile, but their time and memory requirements grow steeply with increasing $n$. Using an SDP solver sdpt3 (\cite{SDPT3}) for Matlab, we were able to solve problems \eqref{eq:Dopt} and \eqref{eq:Dopteq} only for dimensions smaller than $n=4000$ (see Section \ref{sec:Examples} for the specifications of the hardware used).

Finally, for solving $D$-optimal design problems under linear constraints on weights, a promising emerging alternative is a second-order cone programming (SOCP) method developed in \cite{SH14}. Nevertheless, the SOCP methods require very specific software solvers and their actual application for computing optimal designs is technically challenging. Moreover, for the SOCP methods, the degradation of the performance with increasing $n$ is similar to SDP.

Therefore, for computing solutions of problem \eqref{eq:Dopteq}, we decided to construct a specification of the barycentric multiplicative algorithm introduced in \cite{Harman14}. The proposed algorithm has favourable properties similar to standard multiplicative algorithms (see \cite{TM06}, \cite{HP07}, \cite{DPZ08}, \cite{Y10annals}, \cite{Y10csda} for some recent advances in multiplicative methods). More precisely, the barycentric algorithm is very easy to implement and, under mild technical conditions, it has guaranteed monotonic convergence to the optimum. Moreover, the algorithm can be seamlessly combined with stopping rules based on statistical efficiency, as well as with rules for the removal of redundant design points, which yields much more efficient computation. Compared to the vertex direction and the simplicial decomposition methods, the barycentric algorithm does not need to solve a separate optimization problem at each iteration. In contrast to SDP and SOCP, the proposed algorithm has very small memory requirements and can be applied to a large dimension $n$ of the vector of weights.

Naturally, from the point of view of applications, the most important is to find an \emph{exact} experimental design with weights restricted to the set $\{0,\frac{1}{N},\frac{2}{N},\ldots,1\}$, i.e., such that the numbers $N_x=Nw_x$ are integer. An efficient solution of the exact $D$-optimal design problem under \eqref{eq:size} and \eqref{eq:cost} can often be obtained using the corresponding approximate $D$-optimal design, especially if $N$ is large. The simplest method is to ``round down'' an approximate $D$-optimal design $\bw^*$ by replacing the values $w^*_x$ with $\lfloor Nw^*_x \rfloor /N$ for all $x \in \X$, where $\lfloor \cdot \rfloor$ is the floor function. A more efficient solution can be obtained by using an excursion heuristic such as the one proposed in \cite{HBF14}, with the result of the simple rounding taken as an initial feasible solution. Alternatively, it is possible to use the heuristic method from \cite{HF14} based on integer quadratic programming, which utilizes the information matrix of the approximate $D$-optimal design. Another possibility is to use the solution of the approximate problem to identify a small subset of $\X$ that is likely to support the exact $D$-optimal constrained design, and then apply the mixed integer SOCP method described in \cite{SH14}. 

This paper is structured as follows. In Section \ref{sec:Theory}, we provide selected theoretical results about the optimization problem \eqref{eq:Dopteq}, including the so-called equivalence theorem and the rules for the removal of redundant design points. Section \ref{sec:Algorithm} describes the multiplicative barycentric algorithm for computing solutions of \eqref{eq:Dopteq}. Section \ref{sec:Examples} gives a more complex example of a $D$-optimal size-and-cost constrained design and provides a small numerical study exploring the behaviour of the proposed barycentric algorithm. In Section \ref{sec:Notes}, we provide miscellaneous short remarks related to the constrained optimal design problems \eqref{eq:Dopt} and \eqref{eq:Dopteq}. We deferred technical proofs into Section \ref{sec:Appendix}.

Before proceeding, we give a small example to illustrate some aspects of problems \eqref{eq:Dopt} and \eqref{eq:Dopteq}, in particular the fact that for the equality-constrained $D$-optimal design problem \eqref{eq:Dopteq} the values $c_x=1$ play special roles.  

\begin{Example}\label{ex:simple}
Assume that $n=2$, $\bf(1)=(1,0)^{\top}$, $\bf(2)=(1,1)^{\top}$. In this elementary model, it is simple to verify that for a design $\bw=(w_1,w_2)^{\top}$ the criterion of $D$-optimality  is proportional to $\sqrt{w_1w_2}$, and the solutions of both \eqref{eq:Dopt} and \eqref{eq:Dopteq} can be calculated analytically: the solution of \eqref{eq:Dopt} is 
\begin{equation*}
(w_1^*,w_2^*)^\top=
\begin{cases}
(0.5,0.5)^\top &\textrm{if } c_1+c_2\leq 2,\\
\left(\frac{c_1-1}{c_1 - c_2}, \frac{c_2-1}{c_2 - c_1}\right)^\top &\textrm{if } c_1+c_2>2 \text{ and } \frac{1}{2c_1}+\frac{1}{2c_2}>1,\\
\left(\frac{1}{2c_1}, \frac{1}{2c_2}\right)^\top &\textrm{if } \frac{1}{2c_1}+\frac{1}{2c_2}\leq 1,
\end{cases}
\end{equation*}
and the solution of \eqref{eq:Dopteq} is
\begin{equation*}
(w_1^*,w_2^*)^\top=
\begin{cases}
(0.5,0.5)^\top &\textrm{if } c_1=c_2=1,\\
\left(\frac{c_1-1}{c_1 - c_2},\frac{c_2-1}{c_2 - c_1}\right)^\top &\textrm{if } (c_1,c_2)^\top \in ((0,1) \times (1,2)) \cup ((1,2) \times (0,1)).
\end{cases}
\end{equation*}
\begin{figure}[h]
   \captionsetup{width=0.87\textwidth}
   \centering
   \begin{subfigure}[b]{0.4\textwidth}
      \includegraphics[width=\textwidth]{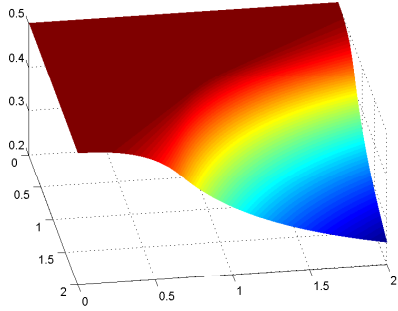}
      \caption{$w_1+w_2 \leq 1, \: c_1w_1+c_2w_2 \leq 1$}
      \label{fig:ineq}
   \end{subfigure}\qquad
   \begin{subfigure}[b]{0.4\textwidth}
      \includegraphics[width=\textwidth]{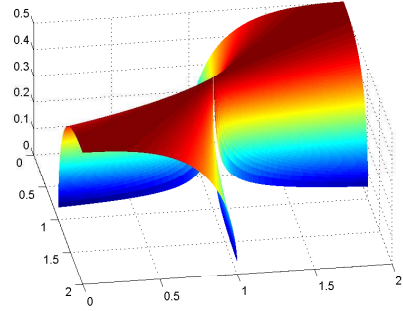}
      \caption{$w_1+w_2 = 1,\: c_1w_1+c_2w_2 = 1$}
      \label{fig:eq}
   \end{subfigure}
 \caption{The values of $\phi$ corresponding to the $D$-optimal designs with $n=2$ design points, regressors $\bf(1)=(1,0)^{\top}$, $\bf(2)=(1,1)^{\top}$ and the costs $c_1,c_2$ varying in $(0,2)$.}\label{fig:c1c2}
\end{figure}

In Figure \ref{fig:c1c2}, we plotted the values of the $D$-criterion in the constrained $D$-optimal designs, as they depend on the costs $c_1$ and $c_2$. For the inequality-constrained problem \eqref{eq:Dopt} illustrated in Figure \ref{fig:ineq}, the optimal criterial values are continuous and non-decreasing for decreasing costs, as expected.

However, the optimal criterial values of the equality-constrained problem \eqref{eq:Dopteq} behave differently. First, in Figure \ref{fig:eq}, the domain of the function is restricted to $\mathcal{C}:=((0,1)\times(1,2)) \cup ((1,2)\times(0,1)) \cup \{(1,1)\}$ because for couples $(c_1,c_2) \in (0,2) \times (0,2) \setminus \mathcal{C}$ there is no feasible solution of \eqref{eq:Dopteq} or the optimal information matrix is singular. Moreover, observe that it is not possible to continuously extend the function in Figure \ref{fig:eq} to the point $(1,1)$, although for $c_1=c_2=1$ the optimization problem \eqref{eq:Dopteq} is meaningful with a unique solution. This phenomenon suggests that the points $x \in \X$ such that $c_x=1$ play a special role. Furthermore, note that in Figure \ref{fig:eq} the optimal criterial value can strictly decrease with decreasing costs.
\end{Example}

\section{Theoretical results for $D$-optimal size-and-cost constrained designs}\label{sec:Theory}

If $c_x \leq 1$ for all $x \in \X$, i.e., if the costs of all trials are ``low'', then every design satisfying the size constraint \eqref{eq:size} satisfies also the cost constraint \eqref{eq:cost}, that is, an optimal solution of \eqref{eq:Dopt} can be found as a solution of \eqref{eq:Doptsize}. Analogously, if $c_x \geq 1$ for all $x \in \X$, that is, if the costs of all trials are ``high'', then every design satisfying the cost constraint \eqref{eq:cost} satisfies also the size constraint \eqref{eq:size}, that is, an optimal solution of \eqref{eq:Dopt} can be found as a solution of \eqref{eq:Doptcost}. Therefore, we can assume that there exist $x_- \in \X$ such that $c_{x_-}<1$ and $x_+ \in \X$ such that $c_{x_+}>1$.

Let $\X_+:=\{x \in \X: c_x > 1\}$, $\X_-:=\{x \in \X: c_x<1\}$, $\X_0:=\{x \in \X: c_x = 1\}$, and let $n_+, n_-, n_0$ be the sizes of these sets. Clearly, our assumptions mean that $\X_+  \neq \emptyset$ as well as $\X_- \neq \emptyset$. For simplicity, in Sections \ref{sec:Theory} and \ref{sec:Algorithm} we will assume that $\X_0 \neq \emptyset$; all results can be modified in a straightforward way for $\X_0=\emptyset$.

Recall that the set of all feasible designs of \eqref{eq:Dopteq} is  denoted $\bbQ^n_+$. We will use the symbol $\bbQ^n_{++}$ to denote the set of all designs $\bw \in \bbQ^n_+$ with all components strictly positive. We will use the symbol $\bbQ^n_r$ to denote the set of all designs $\bw \in \bbQ^n_+$ with a non-singular information matrix $\bM(\bw)$. Note that the regularity assumptions imply $\bbQ^n_{++} \subseteq \bbQ^n_r$.

 Define $\delta_{x_+}:=c_{x_+}-1>0$ for $x_+ \in \X_+$ and $\delta_{x_-}:=1-c_{x_-}>0$ for $x_- \in \X_-$. Consider the design $\bw^{(0)}$ with components
\begin{eqnarray}
w^{(0)}_{x_+} &:=& \tilde{n}^{-1} \sum_{x_- \in \X_-} \frac{\delta_{x_-}}{\delta_{x_+}+\delta_{x_-}}; \hspace{5mm} x_+ \in \X_+, \label{eq:w0plus} \\
w^{(0)}_{x_-} &:=& \tilde{n}^{-1} \sum_{x_+ \in \X_+} \frac{\delta_{x_+}}{\delta_{x_+}+\delta_{x_-}}; \hspace{5mm} x_- \in \X_-, \label{eq:w0minus}\\
w^{(0)}_{x_0} &:=& \tilde{n}^{-1}; \hspace{5mm} x_0 \in \X_0, \label{eq:w0zero}
\end{eqnarray}
where $\tilde{n}:=n_+n_- +n_0$. It is straightforward to verify that $\bw^{(0)}$ is feasible for \eqref{eq:Dopteq}, i.e., $\bbQ^n_+ \neq \emptyset$. Moreover, $\bw^{(0)} \in \bbQ^n_{++}$, that is, $\bw^{(0)} \in \bbQ^n_r$. Hence, the information matrix of the design optimal for \eqref{eq:Dopteq} is non-singular. The strict concavity of $\det^{1/m}(\cdot)$ on the set of all positive definite matrices (\cite{Pukelsheim06}, Section 6.13) guarantees that the optimal information matrix is unique. 

For any $x_+ \in \X_+$, $x_- \in \X_-$ let
\begin{equation*}
  \bq^{(x_+,x_-)}:=\frac{\d_{x_-}}{\d_{x_+}+\d_{x_-}}\be^{(x_+)}+\frac{\d_{x_+}}{\d_{x_+}+\d_{x_-}}\be^{(x_-)}
\end{equation*}
and for any $x_0 \in \X_0$ let $\bq^{(x_0)}=\be^{(x_0)}$, where $\be^{(x)}$, $x \in \X$, are standard unit vectors.
It is simple to show that $\{\bq^{(x_+,x_-)}: x_+ \in \X_+, x_- \in \X_-\} \cup \{\bq^{(x_0)}: x_0 \in \X_0\}$ is the set of all extreme vectors of the polytope $\bbQ^n_+$. In fact, the design $\bw^{(0)}$ defined by \eqref{eq:w0plus}-\eqref{eq:w0zero} is the ``center of mass'' of these extreme vectors if they are assigned equal weights.

For any $\bw \in \bbQ^n_r$, let $\bd(\bw)$ denote the variance (sensitivity) function, which is, in our case, the $n$-dimensional vector with components 
\begin{equation*}
 d_x(\bw) := \bf^{\top}(x) \bM^{-1}(\bw) \bf(x); \: x \in \X.
\end{equation*}
In the standard linear regression model with regressors $\bf(x)$, $x \in \X$, and homoscedastic uncorrelated errors, the value $d_x(\bw)$ is proportional to the variance of the predicted response in the point $x$ under the design $\bw$, see, e.g., Section 2.1 in \cite{FH97} or Section 9.1 in \cite{ADT07}. 

For  any $\bw \in \bbQ^n_r$ and $x_+ \in \X_+$, $x_- \in \X_-$ define the weighted variances
 \begin{equation}\label{eq:Delta}
   \tilde{d}_{x_+ x_-}(\bw) := \frac{\d_{x_+}d_{x_-}(\bw)+\d_{x_-}d_{x_+}(\bw)}{\d_{x_+}+\d_{x_-}}.
 \end{equation}

The form of the extreme vectors of $\bbQ^n_+$ and Theorem 2 from \cite{Harman14} imply the following two theorems. The first one is an ``equivalence theorem'' that characterizes approximate size-and-cost constrained $D$-optimality \eqref{eq:Dopteq}, similarly to the characterization of the standard approximate $D$-optimality, cf.\ Proposition IV.6 in \cite{Pazman86}, Theorem 2.4.1 in \cite{FH97} or Section 9.2 in \cite{ADT07}.

\begin{Theorem}\label{thm:equivalence}
  Let $\bw \in \bbQ^n_r$. Then, $\bw$ is $D$-optimal in $\bbQ^n_+$ if and only if $\max_{x_0 \in \X_0}d_{x_0}(\bw) \leq m$ and 
  \begin{equation*}
    \max_{x_+ \in \X_+}\frac{d_{x_+}(\bw)-m}{\d_{x_+}}+\max_{x_- \in \X_-}\frac{d_{x_-}(\bw)-m}{\d_{x_-}} \leq 0.
   \end{equation*}
\end{Theorem}

It is also possible to formulate an alternative equivalence theorem, analogous to Theorem 4.1.1 in \cite{FH97}. However, the necessary and sufficient condition in Theorem \ref{thm:equivalence} is simpler and more straightforward to verify.

 The second theorem can be used with any sub-optimal feasible design $\bw \in \bbQ^n_r$ to compute a lower bound for its efficiency and delete the points from $\X$ that cannot be in the support of any $D$-optimal size-and-cost constrained design.

\begin{Theorem}\label{thm:effdel}
  Let $\bw \in \bbQ^n_r$, let $\bw^*$ be a design that solves \eqref{eq:Dopteq}, and let
  \begin{equation*}
     \epsilon=\max\left(\max_{x_+ \in \X_+, x_- \in X_-} \tilde{d}_{x_+ x_-}(\bw), \max_{x_0 \in \X_0} d_{x_0}(\bw)\right)-m. 
  \end{equation*}     
  Then, $\mathrm{eff}(\bw|\bw^*) \geq \frac{m}{m+\epsilon}$. Let
  \begin{equation*}
    h_m(\epsilon)=m\left(1+\frac{\epsilon}{2} - \frac{\sqrt{\epsilon(4+\epsilon-4/m)}}{2}\right).
  \end{equation*}    
  Then,\\
  (i) $\max_{x_- \in \X_-} \tilde{d}_{x_+ x_-}(\bw) < h_m(\epsilon)$ for some $x_+ \in \X_+$ implies $w^*_{x_+}=0$.\\
  (ii) $\max_{x_+ \in \X_+} \tilde{d}_{x_+ x_-}(\bw) < h_m(\epsilon)$ for some $x_- \in \X_-$ implies $w^*_{x_-}=0$.\\
  (iii) $d_{x_0}(\bw) < h_m(\epsilon)$ for some $x_0 \in \X_0$ implies $w^*_{x_0}=0$.
\end{Theorem}

The removal of ``redundant'' design points based on Theorem \ref{thm:effdel} can greatly enhance the speed of numerical methods for computing optimal designs, such as the barycentric algorithm derived in the next section.

\section{Barycentric algorithm for computing $D$-optimal size-and-cost constrained designs}\label{sec:Algorithm}

The barycentric algorithm is a multiplicative method proposed in \cite{Harman14} for computing approximate $D$-optimal designs under linear constraints on the vector of weights. The key component of the barycentric algorithm is a formula for (generalised) barycentric coordinates of each $\bw \in \bbQ^n_+$ in a system given by the set of all extreme vectors of $\bbQ^n_+$.

For constraints \eqref{eq:sizeeq} and \eqref{eq:costeq}, the barycentric transformation $\bT^B: \bbQ^n_r \to \bbQ^n_r$ has the form (cf.\ equations (3) and (4) in \cite{Harman14}):
\begin{eqnarray}
  \bT^B(\bw)&=&\frac{1}{m}\mathbf{D}(\bw)\bd(\bw), \text{ where } \label{eq:barygen}\\
  \mathbf{D}(\bw)&=&\left(\sum_{x_+ \in \X_+}\sum_{x_- \in \X_-} \tw_{x_+x_-}(\bw)\bq^{(x_+,x_-)}(\bq^{(x_+,x_-)})^{\top}+\sum_{x_0 \in \X_0} \tw_{x_0}(\bw)\bq^{(x_0)}(\bq^{(x_0)})^{\top}\right). \label{eq:barymat}
\end{eqnarray}
In \eqref{eq:barymat}, the functions $\tw_{x_+x_-}: \bbQ^n_+ \to \bbR$; $x_+ \in \X_+, x_- \in \X_-$, and $\tw_{x_0}:\bbQ^n_+ \to \bbR$; $x_0 \in \X_0$, are the barycentric coordinates, that is, they are non-negative and satisfy
\begin{eqnarray}
  \sum_{x_+ \in \X_+}\sum_{x_- \in \X_-} \tw_{x_+x_-}(\bw)+\sum_{x_0 \in \X_0} \tw_{x_0}(\bw)&=&1, \label{eq:bary1}\\
  \sum_{x_+ \in \X_+}\sum_{x_- \in \X_-} \tw_{x_+x_-}(\bw)\bq^{(x_+,x_-)}+\sum_{x_0 \in \X_0} \tw_{x_0}(\bw)\bq^{(x_0)} &=& \bw \label{eq:baryw}
\end{eqnarray}
for all $\bw \in \bbQ^n_+$. For the general theory from \cite{Harman14} to be applicable, the barycentric coordinates must be chosen such that they are continuous on $\bbQ^n$ and strictly positive for any $\bw \in \bbQ^n_{++}$.

For all $\bw \in \bbQ^n_+$, denote 
\begin{equation*}
S(\bw):=\sum_{x_+ \in \X_+} \delta_{x_+}w_{x_+} = \sum_{x_- \in \X_-} \delta_{x_-}w_{x_-},
\end{equation*}
where the second equality follows directly from \eqref{eq:sizeeq} and \eqref{eq:costeq}.

The barycentric coordinates are not uniquely defined, and not all choices of barycentric coordinates are equally good. It turns out that for problem \eqref{eq:Dopteq} a suitable definition of barycentric coordinates of $\bw \in \bbQ^n_+$ is
\begin{eqnarray}
    \tw_{x_+ x_-}(\bw)&=& 
      \begin{dcases}
        \frac{(\d_{x_+}+\d_{x_-})w_{x_+}w_{x_-}}{S(\bw)} & \text{ if } S(\bw)>0, \label{eq:baryweights_pm} \\
        0 & \text{ if } S(\bw)=0,
      \end{dcases} \\
    \tw_{x_0}(\bw)&=&w_{x_0}, \label{eq:baryweights_0}
\end{eqnarray}
for all $x_+ \in \X_+$, $x_- \in \X_-$, and $x_0 \in \X_0$.

\begin{Proposition}\label{lem:baryweights}
  Let $x_+ \in \X_+$, $x_- \in \X_-$, $x_0 \in \X_0$. The functions $\tw_{x_+x_-}: \bbQ^n_+ \to \bbR$, and $\tw_{x_0}:\bbQ^n_+ \to \bbR$ defined by \eqref{eq:baryweights_pm} and \eqref{eq:baryweights_0} are non-negative, continuous on $\bbQ^n_+$, and positive on $\bbQ^n_{++}$. Moreover, for any $\bw \in \bbQ^n_+$ the functions satisfy \eqref{eq:bary1} and \eqref{eq:baryw}. 
\end{Proposition}

The barycentric algorithm starts with a design $\bw^{(0)} \in \bbQ^n_{++}$ and computes a sequence of designs $\{\bw^{(t)}\}_{t=0}^\infty$ by
\begin{equation*}
  \bw^{(t+1)}=\bT^B(\bw^{(t)}) \text{ for all } t=0,1,2,...
\end{equation*}
until some convergence criterion is satisfied, for instance based on the efficiency bound from Theorem \ref{thm:effdel}. For the practical utility of the resulting algorithm the barycentric coordinates must be chosen such that the transformation $\bT^B$ has a computationally efficient form and guarantees that the sequence $\{\phi(\bw^{(t)})\}_{t=0}^\infty$ converges to the optimal criterial value.
 
 Let us derive the form of the barycentric transformation \eqref{eq:barygen} for any $\bw \in \bbQ^n_{++}$. The diagonal element of the update matrix \eqref{eq:barymat} corresponding to $y_+ \in \X_+$ is
 \begin{eqnarray}\label{eq:D1}
  \left(\mathbf{D}(\bw)\right)_{y_+y_+} &=& \sum_{x_+ \in \X_+}\sum_{x_- \in \X_-} \tw_{x_+x_-}(\bw)(q^{(x_+,x_-)}_{y_+})^2+\sum_{x_0 \in \X_0} \tw_{x_0}(\bw)(q^{(x_0)}_{y_+})^2 \nonumber \\
  &=& \sum_{x_- \in \X_-} \frac{(\d_{y_+}+\d_{x_-})w_{y_+} w_{x_-}}{S(\bw)} \left(\frac{\d_{x_-}}{\d_{y_+}+\d_{x_-}}\right)^2
  = \frac{w_{y_+}}{S(\bw)}\sum_{x_- \in \X_-} \frac{w_{x_-} \d_{x_-}^2}{\d_{y_+}+\d_{x_-}}.
 \end{eqnarray}
  An analogous formula is valid for $y_- \in \X_-$. For $y_+ \in \X_+$ and $y_- \in \X_-$ the element $(y_+,y_-)$ of the update matrix is
  \begin{eqnarray}\label{eq:D2}
  \left(\mathbf{D}(\bw)\right)_{y_+y_-}&=&\sum_{x_+ \in \X_+}\sum_{x_- \in \X_-} \tw_{x_+x_-}(\bw)q^{(x_+,x_-)}_{y_+}q^{(x_+,x_-)}_{y_-}+\sum_{x_0 \in \X_0} \tw_{x_0}(\bw)q^{(x_0)}_{y_+}q^{(x_0)}_{y_-} \nonumber \\
  &=& \frac{(\d_{y_+}+\d_{y_-})w_{y_+}w_{y_-}}{S(\bw)} \frac{\d_{y_+}\d_{y_-}}{\left(\d_{y_+}+\d_{y_-}\right)^2}
  = \frac{w_{y_+}w_{y_-}}{S(\bw)} \frac{\d_{y_+}\d_{y_-}}{\d_{y_+}+\d_{y_-}}.
 \end{eqnarray}
For  $y_0 \in \X_0$, the diagonal element of the update matrix corresponding to $y_0$ is
 \begin{equation}\label{eq:D3}
  \left(\mathbf{D}(\bw)\right)_{y_0y_0} = \sum_{x_+ \in \X_+}\sum_{x_- \in \X_-} \tw_{x_+x_-}(\bw)(q^{(x_+,x_-)}_{y_0})^2+\sum_{x_0 \in \X_0} \tw_{x_0}(\bw)(q^{(x_0)}_{y_0})^2  = w_{y_0}.
 \end{equation}
 It can be easily checked that all other elements of the update matrix are equal to zero. Equalities \eqref{eq:D1}-\eqref{eq:D3} yield the following form of the barycentric updating rule for $\bw \in \bbQ^n_{++}$:
\begin{equation}\label{eq:bary}
 \bT^B(\bw)=\bw \odot \bd^{\pi}(\bw),
\end{equation}
where $\odot$ is the componentwise multiplication and the components of $\bd^{\pi}(\bw)$ are
\begin{eqnarray}
  d^{\pi}_{x_+}(\bw)&=&\frac{\sum_{x_- \in \X_-} w_{x_-} \d_{x_-} \tilde{d}_{x_+ x_-}(\bw)}{mS(\bw)} ; \: x_+ \in \X_+, \label{eq:baryplus}\\
  d^{\pi}_{x_-}(\bw)&=&\frac{\sum_{x_+ \in \X_+} w_{x_+} \d_{x_+} \tilde{d}_{x_+ x_-}(\bw)}{mS(\bw)} ; \: x_- \in \X_-, \label{eq:baryminus}\\
  d^{\pi}_{x_0}(\bw)&=& \frac{d_{x_0}(\bw)}{m}; \: x_0 \in \X_0. \label{eq:baryzero}
\end{eqnarray}
Note that the barycentric transformation uses the numbers $\tilde{d}_{x_+ x_-}(\bw)$, $x_+ \in \X_+, x_- \in \X_-$ and $d_{x_0}(\bw)$, $x_0 \in \X_0$, which can be directly re-used for computing the lower bound on the design efficiency and for the deletion method given in Theorem \ref{thm:effdel}.

Let $\bw^{(0)} \in \bbQ^n_{++}$ be an initial design. Let $\bw^{(t+1)}=\bT^B(\bw^{(t)})$ for $t=0,1,2,...$. Note that $\bM(\bw^{(0)})$ is non-singular. We know from the general theory in \cite{Harman14} that $\{\det(\bM(\bw^{(t)}))\}_{t=0}^\infty$ forms a non-decreasing sequence, i.e., all matrices $\bM(\bw^{(t)})$ are non-singular.  For all $t=0,1,2,...$ and all $x \in \X$ we have $d_x(\bw^{(t)})=\bf^{\top}(x)\bM^{-1}(\bw^{(t)})\bf(x)>0$, which follows from positive definitness of $\bM(\bw^{(t)})$ and from the assumption $\bf(x) \neq 0$ for all $x \in \X$. Hence, the formula for $\bT^B$ implies that all components of all designs $\bw^{(t)}$ are strictly positive.

The general theory in \cite{Harman14} guarantees that the sequence $\{\bM(\bw^{(t)})\}_{t=0}^\infty$ converges to some non-singular matrix $\bM^\infty$, but it does not guarantee that $\bM^\infty$ is the optimal information matrix, i.e., the information matrix of a solution $\bw^*$ of problem \eqref{eq:Dopteq}. However,  it is possible to show that under a mild technical condition the designs $\bw^{(t)}$ converge to the optimum in the sense that their criterial values $\phi(\bw^{(t)})$ converge to the optimal value of \eqref{eq:Dopteq}:

\begin{Theorem}\label{thm:convergence}
 Let $\bw^{(0)} \in \bbQ^n_{++}$ and let $\bw^{(t+1)}=\bT^B(\bw^{(t)})$ for $t=0,1,2,...$. Let $\liminf_{t \to \infty} S(\bw^{(t)})>0$. Then, $\lim_{t \to \infty} \phi(\bw^{(t)})= \phi(\bw^*)$, where $\bw^*$ is any solution of\eqref{eq:Dopteq}.
\end{Theorem}

Technical condition $\liminf_{t \to \infty} S(\bw^{(t)})>0$ is automatically satisfied once $\X_0=\emptyset$. The case $\X_0 \neq \emptyset$ takes place only if $c_x$ is exactly equal to one for some $x \in \X$, which is likely to occur very infrequently in applications. Moreover, even if this is the case, i.e., if $\X_0 \neq \emptyset$, it is reasonable to adopt a conservative approach by slightly increasing the costs $c_{x_0}$, $x_0 \in \X_0$. Alternatively, one can use the following lemma. 
\begin{Lemma}\label{lem:guarantee}
 Let $\bw^{(0)} \in \bbQ^n_{++}$ and let $\bw^{(t+1)}=\bT^B(\bw^{(t)})$ for $t=0,1,2,...$.  Let $v_0 := \max\{\phi(\bw): \bw \geq \0_n, \sum_{x_0 \in \X_0} w_{x_0}=1\}$, that is, $v_0$ is the optimal value of the standard problem of approximate $D$-optimality on $\X_0$. Assume that $\phi(\bw^{(s)})>v_0$ for some $s \in \{0,1,2,...\}$. Then, $\liminf_{t \to \infty} S(\bw^{(t)})>0$.
\end{Lemma}
In most cases, the value $v_0$ from Lemma \eqref{lem:guarantee} is so small, that $\phi(\bw^{(s)})>v_0$ is satisfied already for the initial design $\bw^{(0)}$. In such cases, the convergence of the barycentric algorithm is guaranteed from the outset. 

\section{Numerical study}\label{sec:Examples}

Assume the full quadratic linear regression model with homoscedastic uncorrelated observations on a $101 \times 101$ equidistant rectangular grid in the square $[0,1] \times [0,1]$. For this model, the observations $y$ satisfy
\begin{equation}\label{eq:quadmod}
 E(y)=\theta_1+\theta_2r_1(x)+\theta_3r_2(x)+ \theta_4r^2_1(x)+\theta_5r^2_2(x)+ \theta_6r_1(x)r_2(x), 
\end{equation} 
where $r_1(x)$ and $r_2(x)$ transform the index $x \in \X=\{1,2,...,101^2\}$ into two coordinates in $[0,1]$ by formulas $r_1(x)=\lfloor (x-1)/101 \rfloor /100$, and $r_2(x)=((x-1) \: \mathrm{mod} \: 101)/100$. That is, the model has $m=6$ unknown parameters $\theta_1,...,\theta_6$ and the regressors are given by
\begin{equation*}
 \bf(x)=(1,r_1(x),r_2(x),r_1^2(x),r_2^2(x),r_1(x)r_2(x))^{\top}.
\end{equation*}
The costs were chosen to be $c_x=0.1+6r_1(x)+r_2(x)$ for all $x \in \X$. Thus, the sizes of the partitions $\X_+$, $\X_-$, and $\X_0$ are $n_+=9465$, $n_-=720$, and $n_0=16$, respectively. Every $16$ iterations, we used Theorem \ref{thm:effdel}, parts (i)-(iii), to remove redundant design points.

 Figures \ref{fig:quad99}, \ref{fig:quad999}, and \ref{fig:quad9999} illustrate the designs and the areas of deleted design points at the moments when the algorithm reached efficiencies $0.99$, $0.999$ and $0.9999$. Figure \ref{fig:quadsum} shows the time-dependence of the iteration number and the number of non-deleted design points. Note that as the size of the design space shrinks, the speed of the computation (measured by the number of iterations) increases. 

\begin{figure}[h]
  \captionsetup{width=0.9\textwidth}
  \centering
  \begin{subfigure}[b]{0.45\textwidth}
     \includegraphics[width=\textwidth]{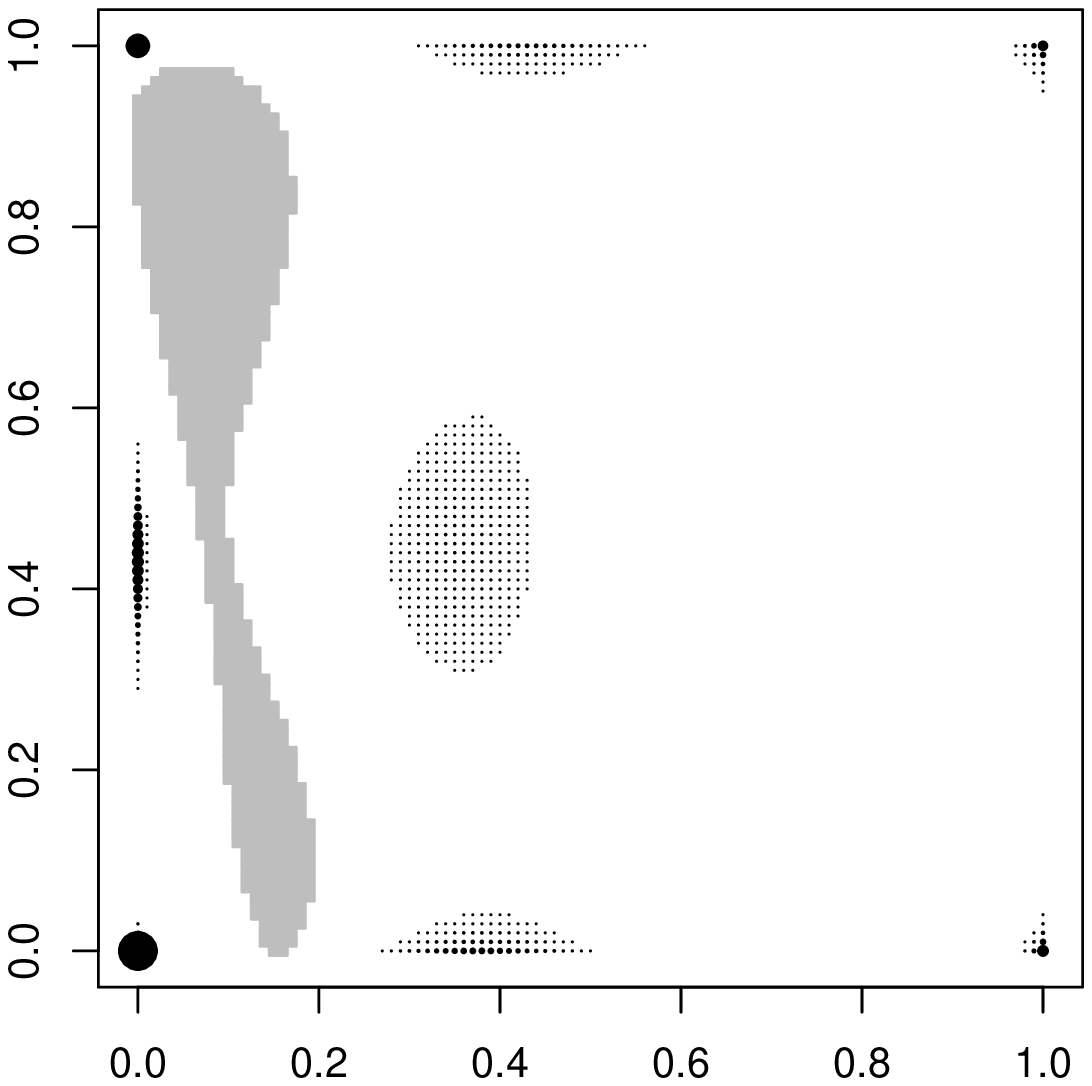}
     \caption{$\mathrm{eff}\geq 0.99$}\label{fig:quad99}
  \end{subfigure}
  \begin{subfigure}[b]{0.45\textwidth}
     \includegraphics[width=\textwidth]{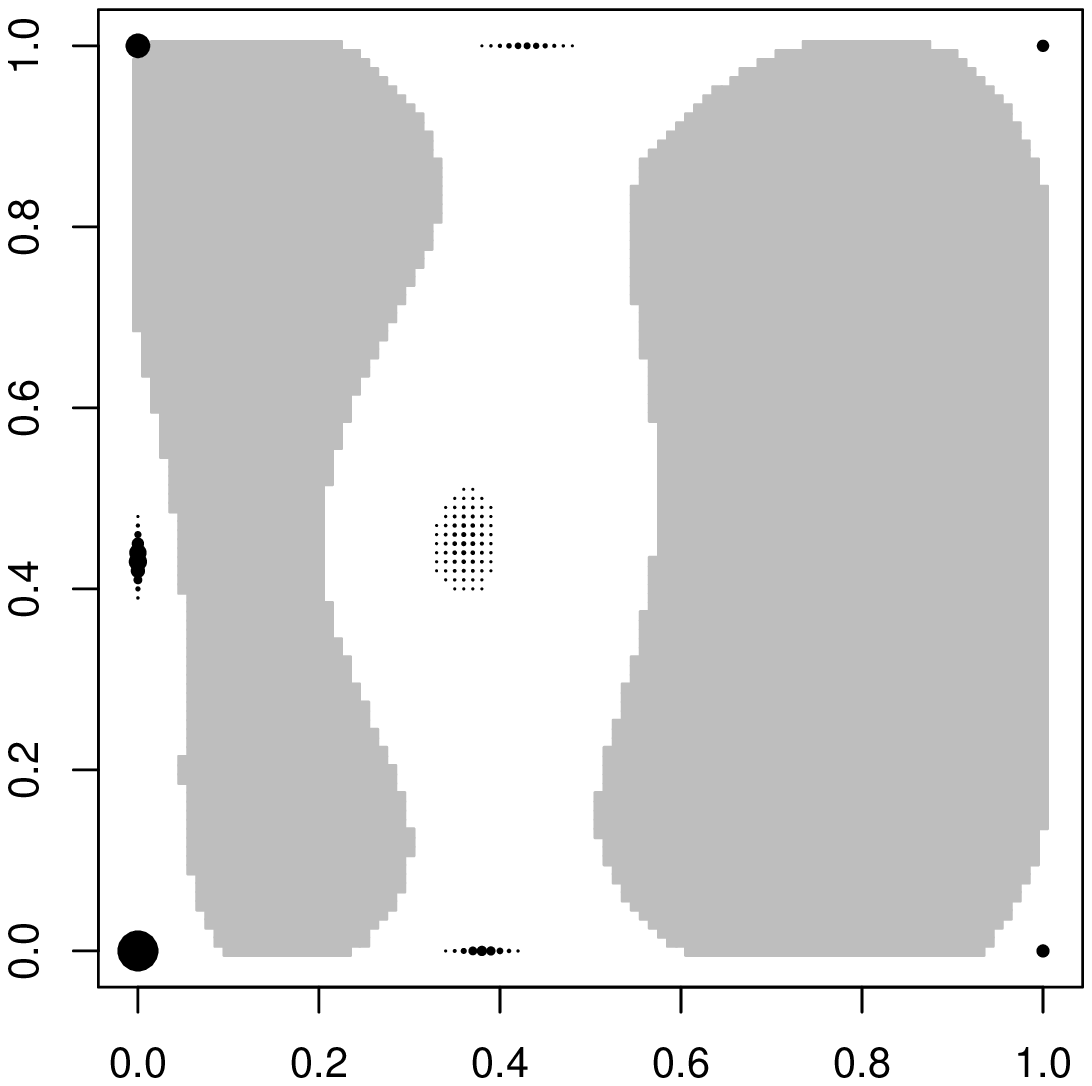}
     \caption{$\mathrm{eff} \geq 0.999$}\label{fig:quad999}
  \end{subfigure}\\
  \begin{subfigure}[b]{0.45\textwidth}
     \includegraphics[width=\textwidth]{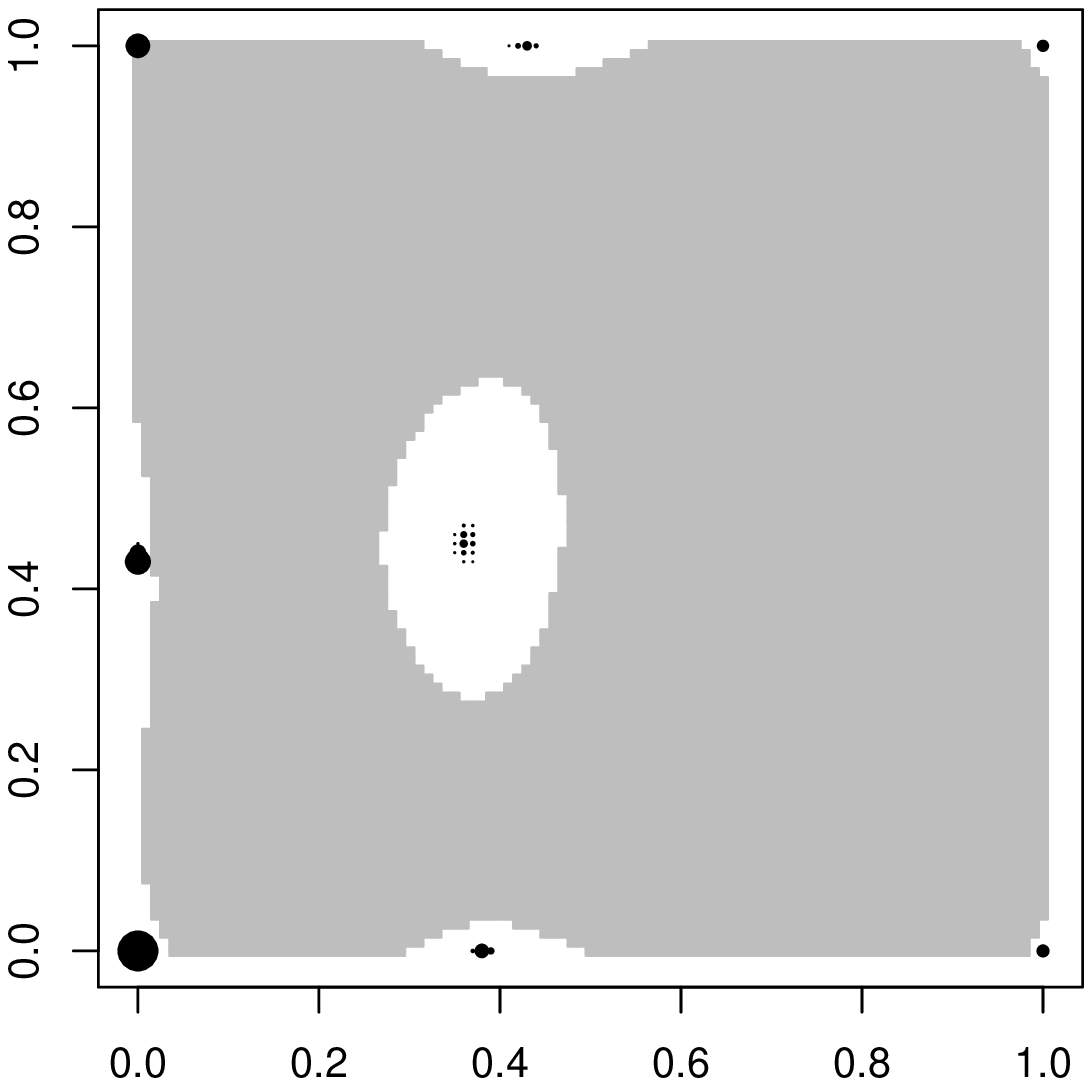}
     \caption{$\mathrm{eff} \geq 0.9999$}\label{fig:quad9999}
       \end{subfigure}
     \begin{subfigure}[b]{0.45\textwidth}
     \includegraphics[width=\textwidth]{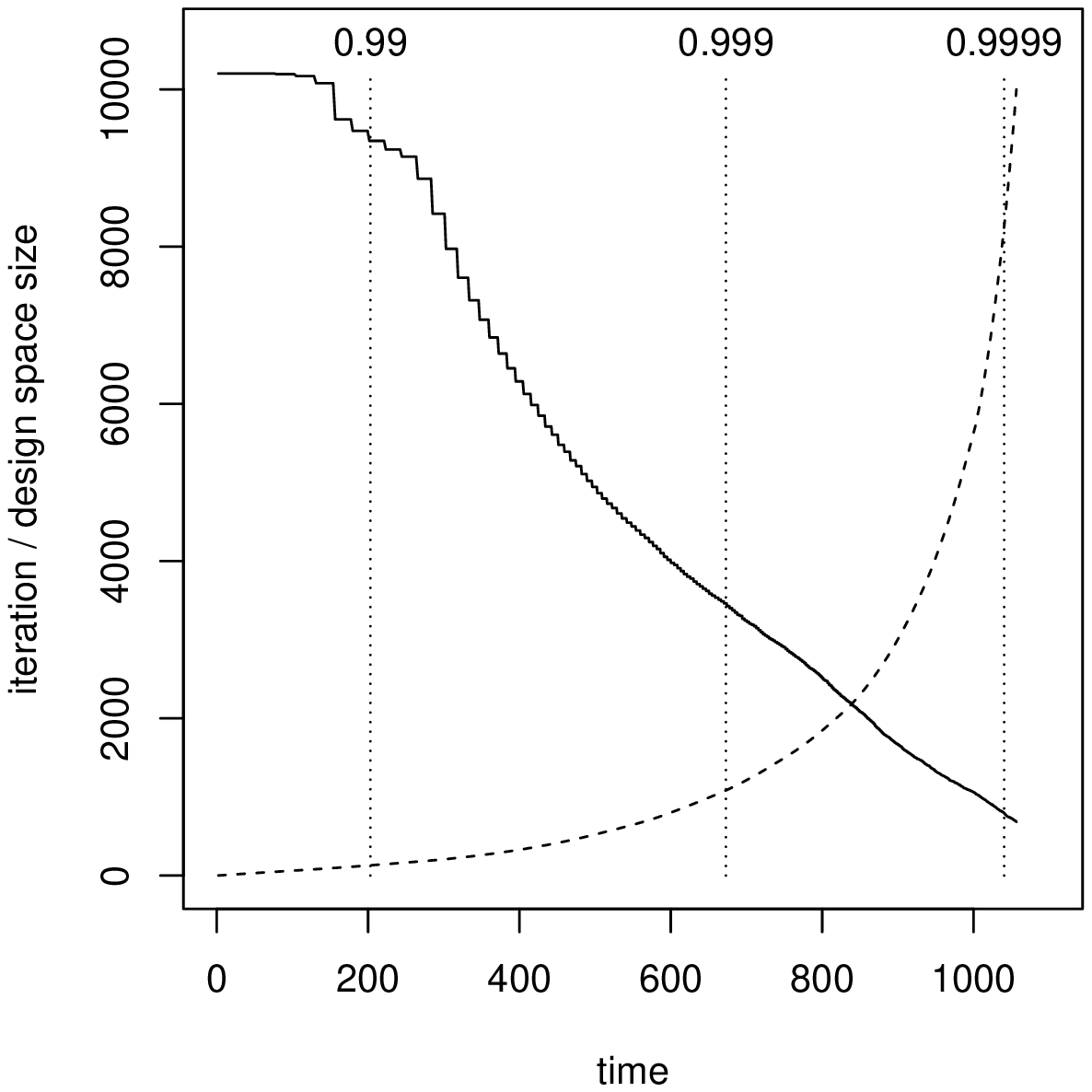}
     \caption{}\label{fig:quadsum}
  \end{subfigure}
  \caption{Figures \ref{fig:quad99}, \ref{fig:quad999}, and \ref{fig:quad9999} visualize the designs constructed using the barycentric algorithm for the quadratic regression model \eqref{eq:quadmod}. The weights are denoted by black dots with areas proportional to their numeric values. The gray regions denote the points of the original design space that have been removed by the deletion rules from Theorem \ref{thm:effdel}. Figure \ref{fig:quadsum} shows the iteration number (dashed line) and the number of residual design points (solid line) as they depend on time in seconds. The vertical lines (dotted) denote the time moments when the efficiencies $0.99$, $0.999$ and $0.9999$ have been achieved, which corresponds to Figures \ref{fig:quad99}, \ref{fig:quad999}, and \ref{fig:quad9999}, respectively.} \label{fig:quad}
\end{figure}

To obtain more general numerical results, we generated random instances of problem \eqref{eq:Dopteq} with the aim to give statistical information about the speed of computation of the barycentric algorithm. Clearly, the execution time can be strongly influenced by the software and the hardware used (we used the Matlab computing environment on 64 bit Windows 7 system running an Intel Core i3-4000M CPU processor at $2,40$ GHz with 4 GB of RAM). Therefore, we also exhibit results about the numbers of iterations, which depend only on the computational method itself.

More specifically, we run the barycentric algorithm $1000$ times for various combinations of parameters $p_0:=n_0/n$, $p_{+-}:=n_+/n_-$ and $l$, where $l$ is the number of iterations between successive applications of the deletion method based on Theorem \ref{thm:effdel}, parts (i)-(iii). In each simulation, we varied one of the parameters $p_0 \in \{0,0.25,0.5,0.75,1\}$, $p_{+-} \in \{0.1,0.3,0.5,0.7,0.9\}$, or $l \in \{1,4,16,64,\infty\}$, keeping all other parameters fixed (the value $l=\infty$ means that the deletion of redundant design points has not been performed at all). The size of the design space and the number of model parameters were always the same: $n=600$ and $m=4$. We did not vary the values $n$ and $m$ because the change of the performance of the algorithm with respect to $n$ and $m$ is analogous to the problem analysed in Section 5 in \cite{Harman14}.

For each triple $p_0,p_{+-},l$, we generated $n_+=\lfloor (1-p_0)p_{+-}n \rfloor$ costs independently from the shifted exponential distribution $\mathrm{Exp}(1)+1$, and $n_-=\lfloor (1-p_0)(1-p_{+-})n \rfloor$ costs independently from the uniform distribution on $(0,1)$. Remaining $n-n_+-n_- \approx np_0$ costs were set to $1$. Regressors $\mathbf{f}(x) \in \bbR^m, x \in \X$, were sampled independently from $\mathrm{N}_m(\mathbf{0}_m,\mathbf{I}_m)$.

The barycentric algorithm started its iterative computation from the initial design $\bw^{(0)}$ defined by \eqref{eq:w0plus}-\eqref{eq:w0zero}. In every step, the current design was updated according to \eqref{eq:bary}-\eqref{eq:baryzero}.

After each successful application of the deletion method, we had to ``re-normalize'' the design to satisfy the constraints of  \eqref{eq:Dopteq}. A natural re-normalization is to set $w_{x_+}=h_+w_{x_+}$ for all $ x_+ \in \X_+$, $w_{x_-}=h_-w_{x_-}$ for all $ x_- \in \X_-$, and $w_{x_0}=h_0w_{x_0}$ for all $ x_0 \in \X_0$, where $h_+$, $h_-$ and $h_0$ are suitably chosen positive constants.  

Let $\bw \geq \0_n$ be a fixed design with $0<s:=\sum_x w_x \leq 1$ and $0<\sum_x c_xw_x \leq 1$. Let $s_+:=\sum_{x_+} w_{x_+}$, $s_-:=\sum_{x_-} w_{x_-}$, and $s_0:=\sum_{x_0} w_{x_0}$. Let $s^\delta_+:=\sum_{x_+} \delta_{x_+} w_{x_+}$, and let $s^\delta_-:=\sum_{x_-} \delta_{x_-} w_{x_-}$. 

Assume that $s_+, s_- > 0$, which implies $s^\delta_+,s^\delta_- > 0$. For the requirement that the re-normalized design should satisfy both \eqref{eq:sizeeq} and \eqref{eq:costeq}, the following linear equalities must hold
\begin{eqnarray}
h_+ s_+ + h_- s_- + h_0 s_0 &=& 1, \label{eq:renorm1}\\
h_+ (s_+ + s^\delta_+) + h_- (s_- - s^\delta_-) + h_0 {s}_0 &=& 1. \label{eq:renorm2}
\end{eqnarray}

If $s_0=0$, then \eqref{eq:renorm1} and \eqref{eq:renorm2} give $h_+=s^\delta_-/(s_+ s^\delta_- + s_- s^\delta_+)$, $h_-=s^\delta_+/(s_+ s^\delta_- + s_- s^\delta_+)$ and $h_0$ can be arbitrary. If $s_0>0$, equalities \eqref{eq:renorm1} and \eqref{eq:renorm2} do not uniquely determine any of the re-normalization factors $h_+,h_-,h_0$. Therefore, motivated by keeping the ratio of the weights of $\X_+ \cup \X_-$ and $\X_0$ the same before and after the re-normalization, we can demand equality
\begin{equation}
\frac{s_+ + s_-}{s_0}=\frac{h_+s_+ + h_-s_-}{h_0s_0}, \label{eq:renorm3}
\end{equation} 
which is linear in $h_+,h_-,h_0$. The solution of the linear system \eqref{eq:renorm1}-\eqref{eq:renorm3} is
\begin{eqnarray*}
 h_+=\frac{s^\delta_-(s_+ + s_-)}{s(s_+ s^\delta_- + s_- s^\delta_+)}, \: h_-=\frac{s^\delta_+(s_+ + s_-)}{s(s_+ s^\delta_- + s_- s^\delta_+)}, \: h_0=\frac{1}{s}.
\end{eqnarray*}
 
 In case $s_+ = s_- =s^\delta_+=s^\delta_-=0$ we must have $s_0>0$, which means that we can simply set $h_0=1/s_0$, and $h_+,h_-$ can be arbitrary. In this case the barycentric algorithm is reduced to the standard multiplicative algorithm without the cost constraint. 
 
The required minimal efficiency was set to $0.99999$, which means that we stopped the algorithm once this lower bound has been reached by the actual design (cf.\ Theorem \ref{thm:effdel}). We remark that the algorithm converged in all $30000$ simulated problems.

The results in the form of boxplots are exhibited in Figures \ref{fig:iter} and \ref{fig:time}. The results indicate that the problem is computationally more demanding for greater values of $\tilde{n}=n_0+n_+n_-$. Thus, with fixed $n=n_0+n_++n_-$, we can generally expect a longer computation time (and, to a lesser extent, a higher number of iterations) for $n_0=0$ and $n_+ \approx n_-$, i.e., for $p_0=0$ and $p_{+-}=0.5$. The numerical results demonstrate that removal of redundant design points can decrease the computation time by an order of magnitude.

\begin{figure}[h]
  \captionsetup{width=0.9\textwidth}
  \centering
  \begin{subfigure}[b]{0.3\textwidth}
     \includegraphics[width=\textwidth]{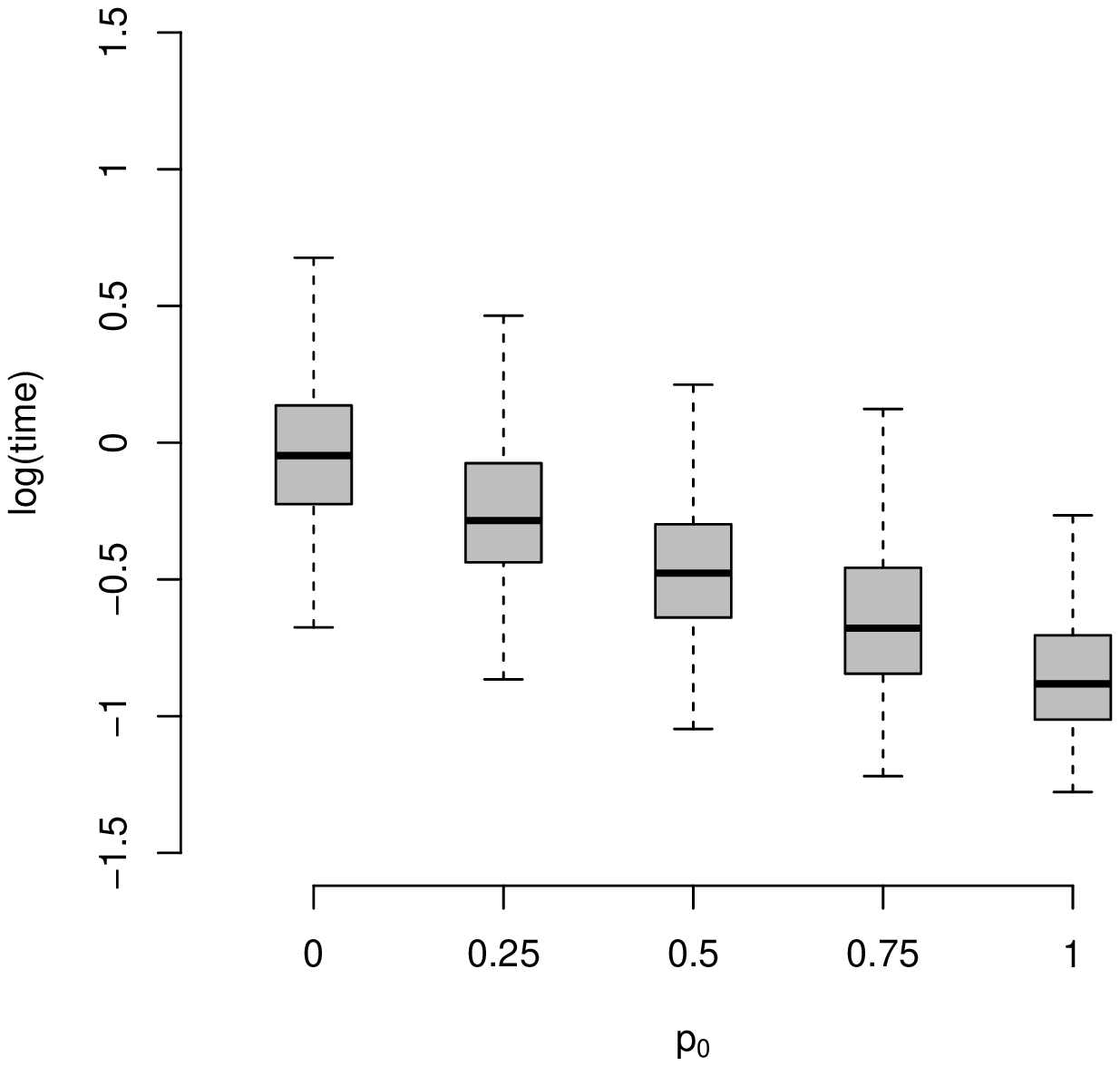}
     \caption{}\label{fig:timep0}
  \end{subfigure}
  \begin{subfigure}[b]{0.3\textwidth}
     \includegraphics[width=\textwidth]{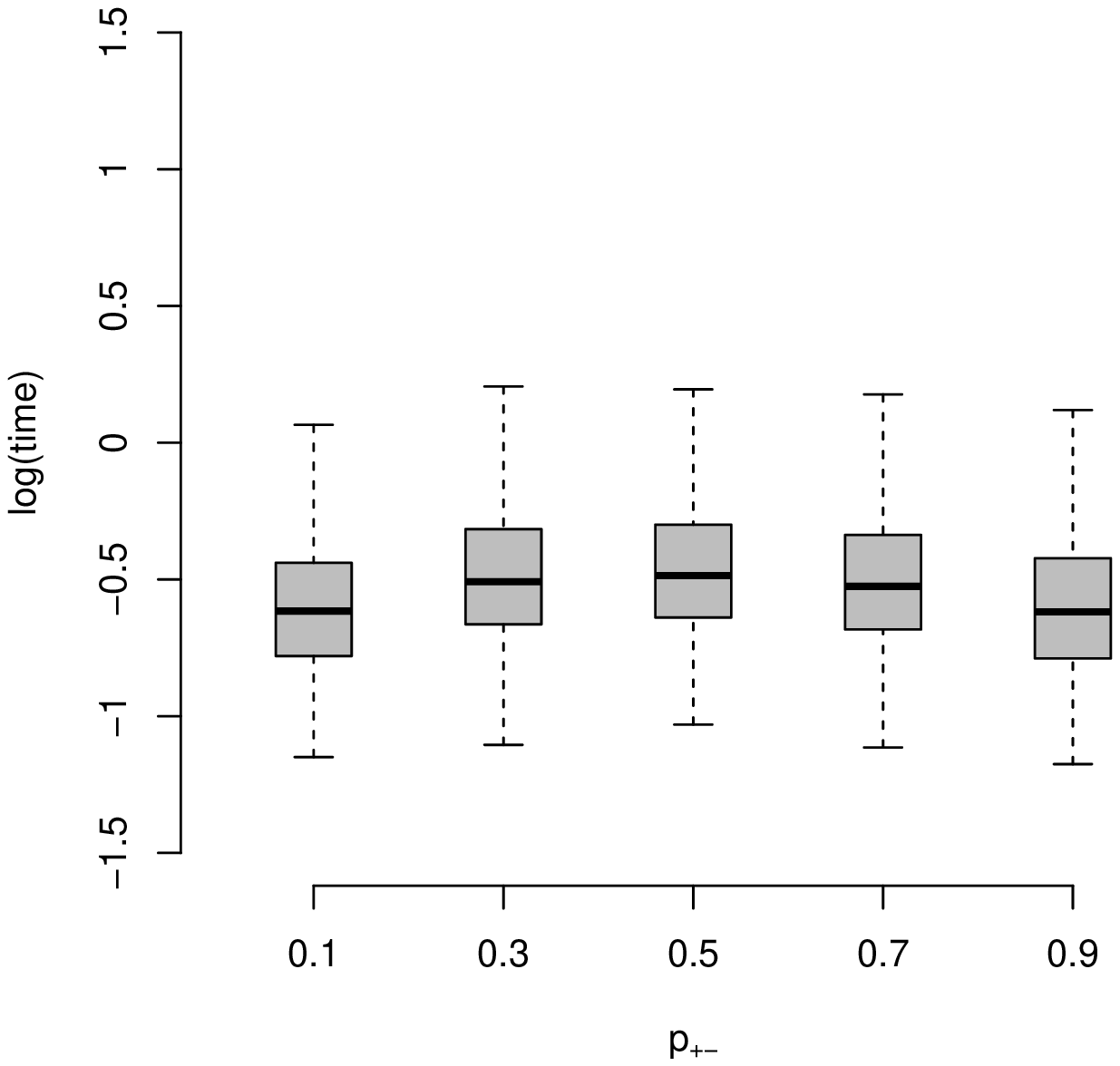}
     \caption{}\label{fig:timep_pl}
  \end{subfigure}
  \begin{subfigure}[b]{0.3\textwidth}
     \includegraphics[width=\textwidth]{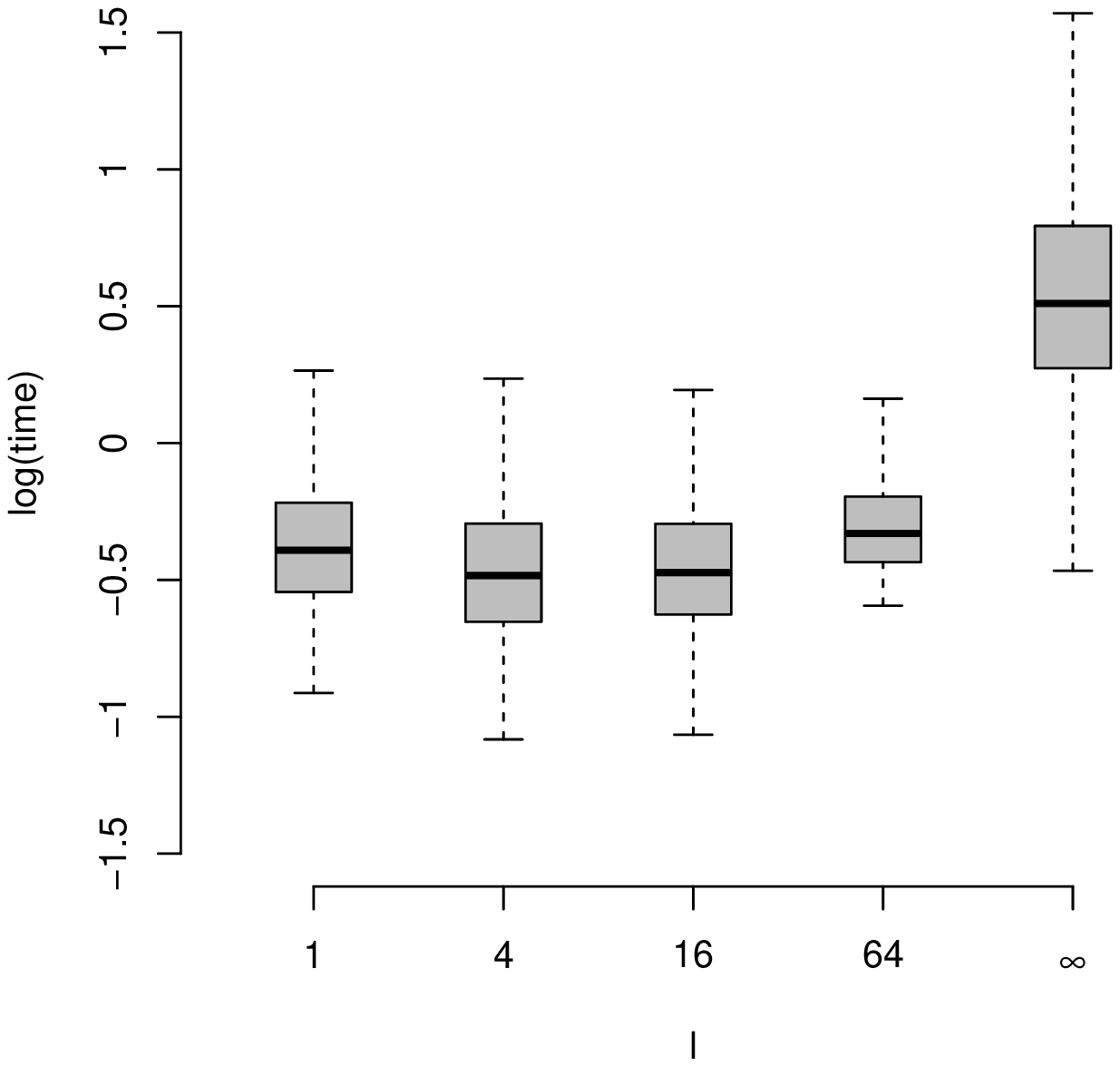}
     \caption{}\label{fig:timel}
  \end{subfigure}
  \caption{The decimal logarithm of the computation time (in seconds) of the barycentric algorithm necessary to achieve the efficiency of $0.99999$. Each boxplot is based on $1000$ randomly generated problems of the type \eqref{eq:Dopteq} with $n=600$ and $m=4$. Figure \ref{fig:timep0}: $p_{+-}=0.5$, $l=16$, and $p_0$ varies in $\{0,0.25,0.5,0.75,1\}$. Figure \ref{fig:timep_pl}: $p_0=0.5$, $l=16$, and $p_{+-}$ varies in $\{0.1,0.3,0.5,0.7,0.9\}$. Figure \ref{fig:timel}: $p_0=0.5$, $p_{+-}=0.5$, $l$ varies in $\{1,4,16,64,\infty\}$. The value $l=\infty$ means that no deletions were performed. See the main text for more details.}\label{fig:time}
\end{figure}

\begin{figure}[h]
  \captionsetup{width=0.9\textwidth}
  \centering
  \begin{subfigure}[b]{0.3\textwidth}
     \includegraphics[width=\textwidth]{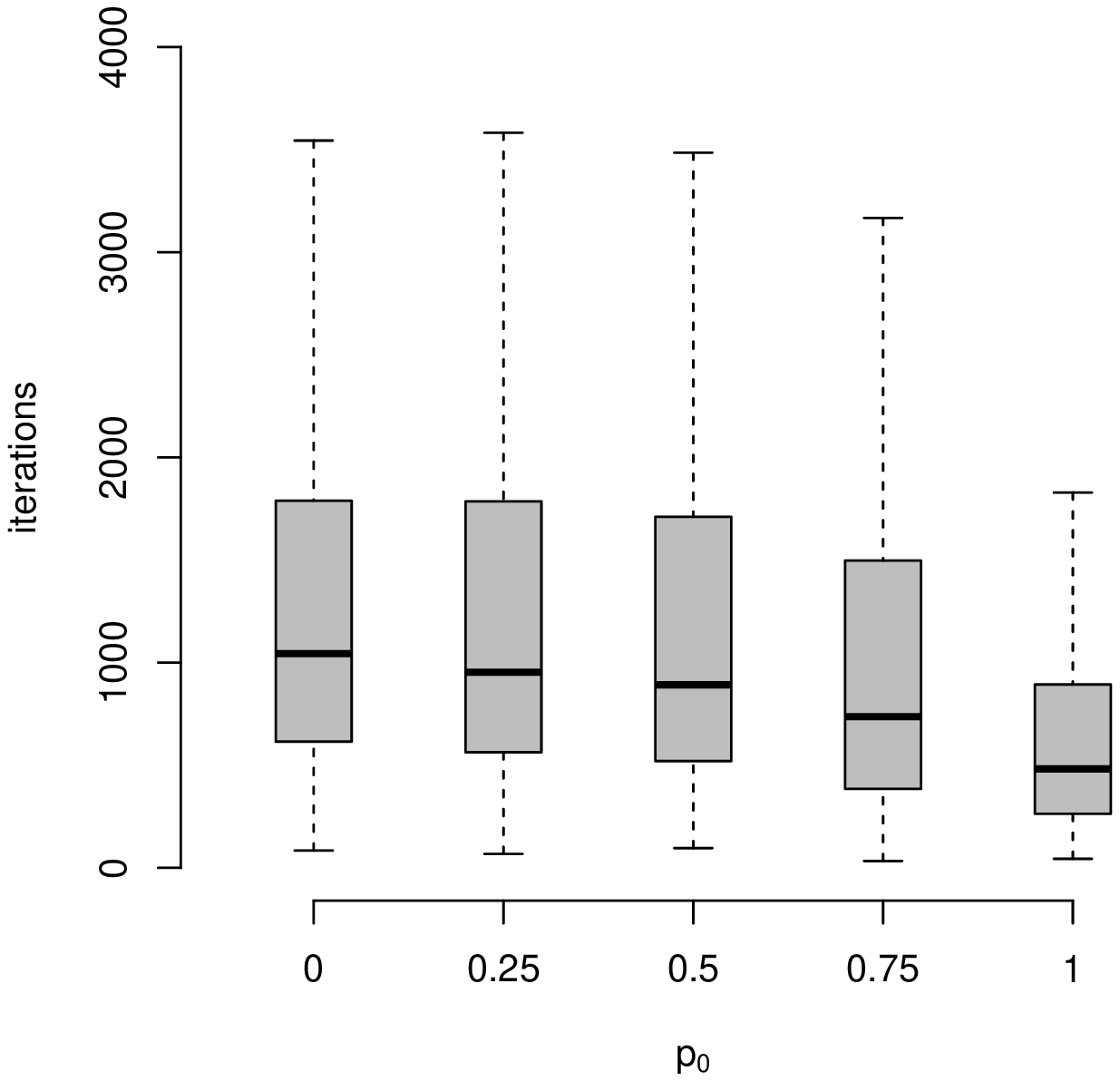}
     \caption{}\label{fig:iterp0}
  \end{subfigure}
  \begin{subfigure}[b]{0.3\textwidth}
     \includegraphics[width=\textwidth]{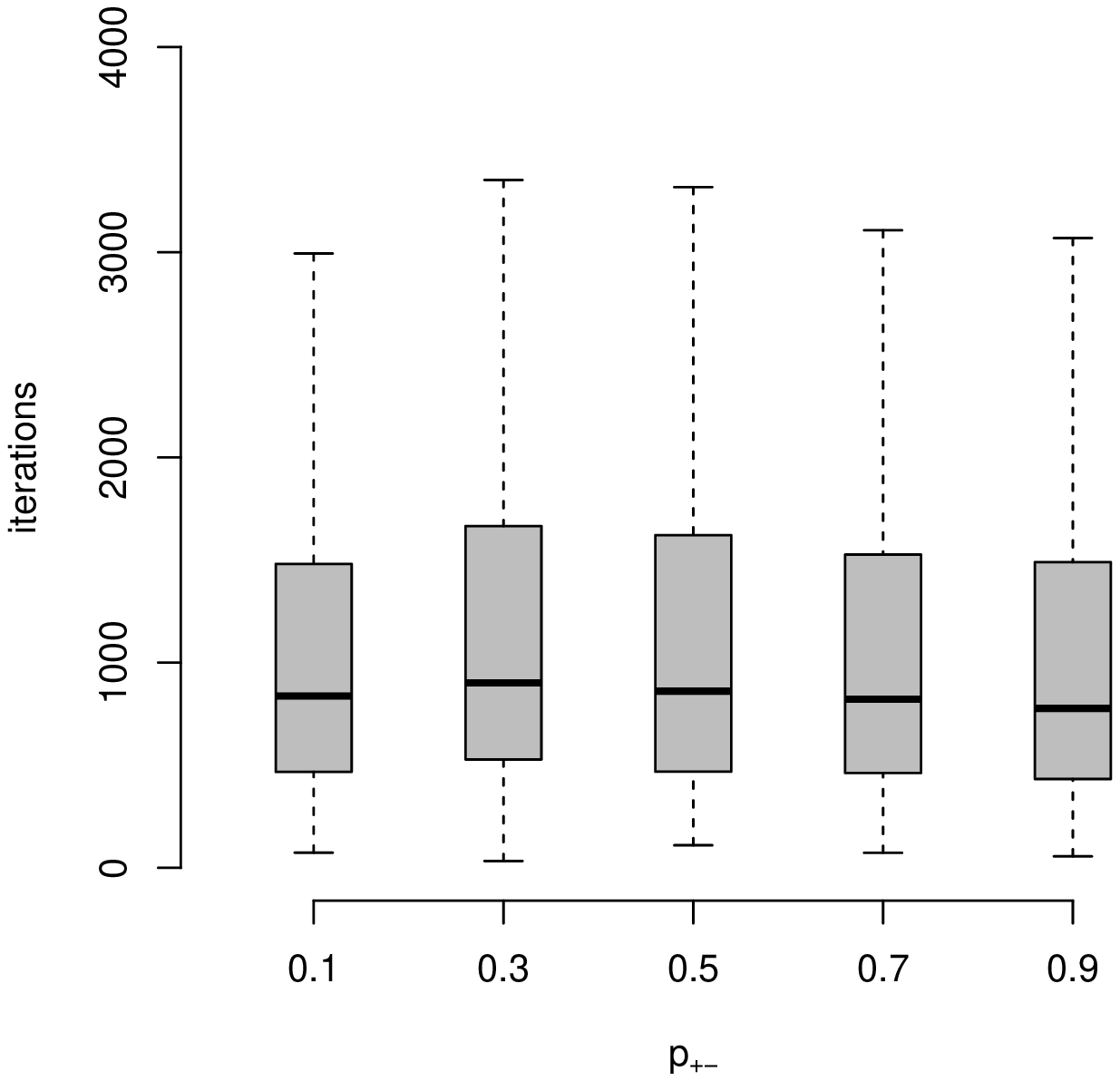}
     \caption{}\label{fig:iterp_pl}
  \end{subfigure}
  \begin{subfigure}[b]{0.3\textwidth}
     \includegraphics[width=\textwidth]{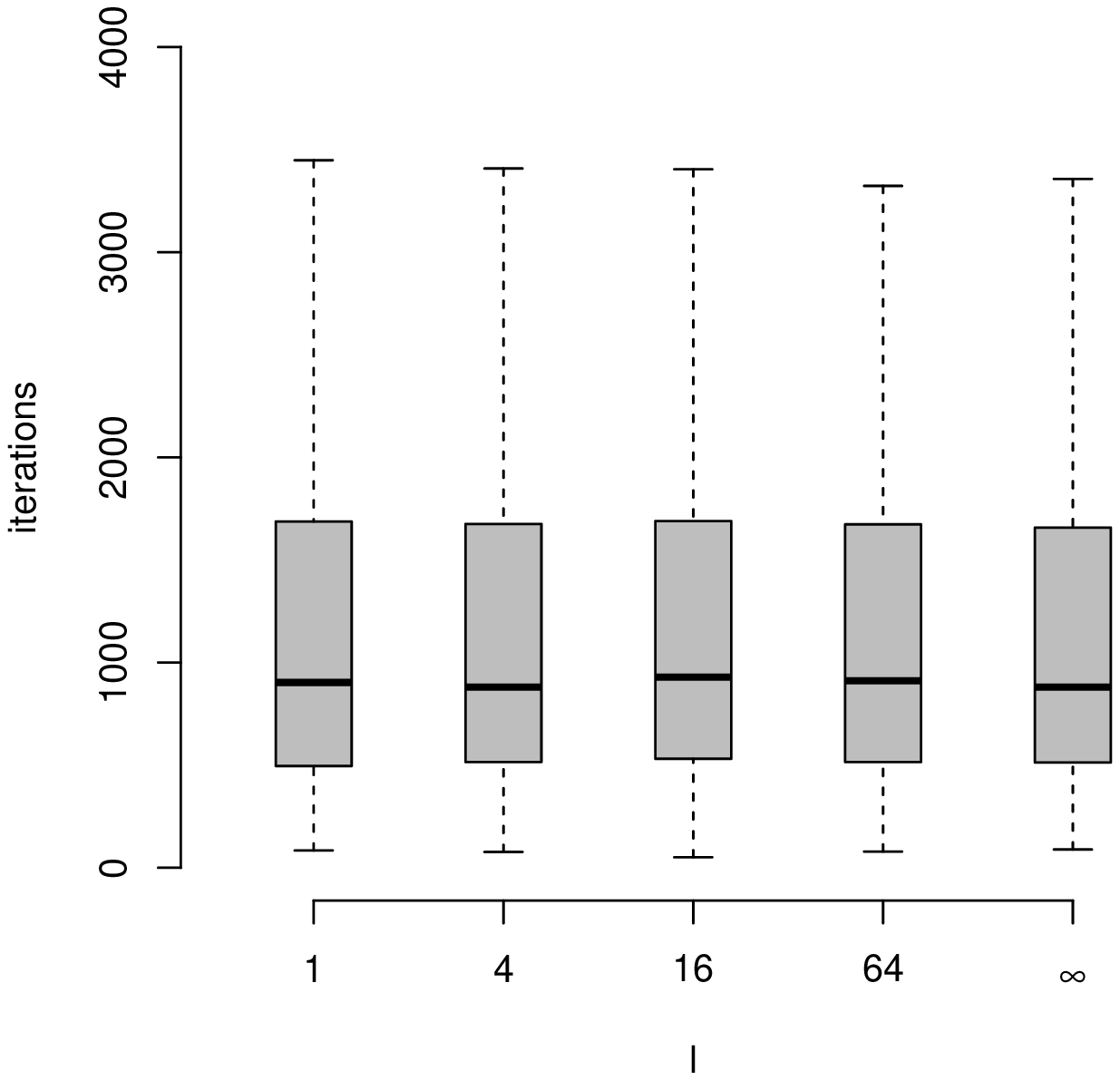}
     \caption{}\label{fig:iterl}
  \end{subfigure}
  \caption{The number of iterations of the barycentric algorithm necessary to achieve the efficiency of $0.99999$. Each boxplot is based on $1000$ randomly generated problems of the type \eqref{eq:Dopteq} with $n=600$ and $m=4$. Figure \ref{fig:iterp0}: $p_{+-}=0.5$, $l=16$, and $p_0$ varies in $\{0,0.25,0.5,0.75,1\}$. Figure \ref{fig:iterp_pl}: $p_0=0.5$, $l=16$, and $p_{+-}$ varies in $\{0.1,0.3,0.5,0.7,0.9\}$. Figure \ref{fig:iterl}: $p_0=0.5$, $p_{+-}=0.5$, $l$ varies in $\{1,4,16,64,\infty\}$. The value $l=\infty$ means that no deletions were performed. See the main text for more details.} \label{fig:iter}
\end{figure}

The Matlab code that implements the barycentric algorithm with the deletion method is available at: \verb!www.iam.fmph.uniba.sk/design/! .

\section{Additional remarks}\label{sec:Notes}

\subsection{A pair of general linear constraints}
Consider the $D$-optimal design problem 
\begin{equation}\label{eq:Doptgen}
\tilde{\bw}^* \in \mathrm{argmax}\left\{\phi(\tilde{\bw}): \tilde{\bw} \geq \0_n, \sum_x c^{(1)}_x\tilde{w}_x \leq 1, \sum_x c^{(2)}_x\tilde{w}_x \leq 1 \right\},
\end{equation}
where $c^{(1)}_x, c^{(2)}_x>0$, for all $x \in \X$, are given constants. Let $\tilde{\bf}(x)$, $x \in \X$, denote the regressors. It turns out that problem \eqref{eq:Doptgen} can be transformed to \eqref{eq:Dopt} in an analogous way as the single-cost constrained problem \eqref{eq:Doptcost} can be transformed to the standard problem \eqref{eq:Doptsize}. Specifically, it is possible to use transformations \begin{equation*}
 w_x=c^{(1)}_x\tilde{w}_x, \: c_x=\frac{c^{(2)}_x}{c^{(1)}_x}, \: \bf(x)=\frac{\tilde{\bf}(x)}{\sqrt{c^{(1)}_x}}; \: x \in \X.
\end{equation*} 
 Similarly, problem \eqref{eq:Doptgen} but with the inequality constraints replaced with equality constraints can be transformed to \eqref{eq:Dopteq}.

\subsection{Alternative application of the cost constraint}
Equality \eqref{eq:costeq} can be used to guarantee a fixed value of $\tilde{\phi}(\bw)=\mathrm{tr}(\Sigma \bM(\bw))$, where $\Sigma$ is a known positive definite $m \times m$ matrix. Indeed, $\tilde{\phi}(\bw)=v>0$ is equivalent to \eqref{eq:costeq} with $c_x=v^{-1}\bf^{\top}(x)\Sigma\bf(x)>0$. Criteria of this form include the Kiefer's $\phi_1$ criterion (Section 6.5 in \cite{Pukelsheim06}) as well as a criterion that can be used for model testing (equation (13) in \cite{FK86}; cf.\ Example 3 in \cite{CF95}).

\subsection{Relations to stratified $D$-optimality}
Let $c_{x_+} \equiv c_+ >1$ for all $x_+ \in \X_+$, $c_{x_-} \equiv c_- \in (0,1)$ for all $x_- \in \X_-$ and let $\X_0 = \emptyset$. Define $s_+(\bw):=\sum_{x_+} w_{x_+}$ and $s_-(\bw):=\sum_{x_-} w_{x_-}$. If $\bw$ is feasible for \eqref{eq:Dopteq}, then
$s_+(\bw) + s_-(\bw)=1$ and $c_+ s_+(\bw) + c_- s_-(\bw)=1$. Hence the values
\begin{equation*}
  s_+:=s_+(\bw)=\frac{1-c_-}{c_+-c_-}, \: s_-:=s_-(\bw)=\frac{c_+-1}{c_+-c_-}
\end{equation*}  
 do not depend on $\bw$. Therefore, for this very specific size-and-cost constrained case, the set of designs feasible for \eqref{eq:Dopteq} is the same as the set of stratified designs with partitions $\X_+$, $\X_-$ and weights $s_+, s_-$, see \cite{Harman14}. That is, in this situation the size-and-cost constrained $D$-optimality coincides with the stratified $D$-optimality, including the equivalence theorem, the deletion rules, and also the barycentric algorithm.

\subsection{Matrix form of the barycentric algorithm}
For computations, it may be useful to rewrite the  barycentric transformations \eqref{eq:baryplus}-\eqref{eq:baryzero} to the following form. For vectors $\mathbf{a} \in \bbR^{n_+}$ and $\mathbf{b} \in \bbR^{n_-}$ let $\mathbf{a} \oplus \mathbf{b}^{\top}$ be the $n_+ \times n_-$ matrix with components $a_{x_+} + b_{x_-}$. This operation can be implemented as a stand-alone function or using the Kronecker multiplication. For any vector $\mathbf{g} \in \bbR^n$, let $\mathbf{g}_+ \in \bbR^{n_+}$, $\mathbf{g}_-  \in \bbR^{n_-}$, and $\mathbf{g}_0  \in \bbR^{n_0}$ denote the sub-vectors of $\mathbf{g}$ corresponding to $x_+ \in \X_+$, $x_- \in \X_-$, and $x_0 \in \X_0$, respectively. Let $\bw \in  \bbQ^n_+$ be a feasible design, let $\delta:=(\delta_1,...,\delta_n)^\top$ and $\1_k:=(1,...,1)^\top \in \bbR^k$ for any $k \in \bbN$. The $n_+ \times n_-$ matrix $\Delta$ with components defined in \eqref{eq:Delta} is 
\begin{equation*}
\Delta=\left[(\bd_+(\bw) \oslash \delta_+) \oplus (\bd_-(\bw) \oslash \delta_-)^\top \right] \oslash \left[(\mathbf{1}_{n_+} \oslash \delta_+) \oplus (\mathbf{1}_{n_-} \oslash \delta_-)^\top \right],
\end{equation*}
where $\oslash$ denotes the componentwise division. Then, the barycentric transformations \eqref{eq:baryplus}-\eqref{eq:baryzero} can be written in the form 
\begin{eqnarray}
\bT^B_+(\bw) &=& \frac{1}{m S(\bw)} \bw_+ \odot [\Delta (\bw_- \odot \delta_-)], 
\label{eq:matrix1}\\
\bT^B_-(\bw) &=& \frac{1}{m S(\bw)} \bw_- \odot [\Delta^\top (\bw_+ \odot \delta_+)], \label{eq:matrix2}\\
\bT^B_0(\bw) &=& \frac{1}{m} \bw_0 \odot \bd_0(\bw), \label{eq:matrix3}
\end{eqnarray}
where $\odot$ denotes the componentwise multiplication. In matrix-based software such as Matlab or R, computations \eqref{eq:matrix1}-\eqref{eq:matrix3} can be performed very efficiently. 

\subsection{Increasing the speed of computations}
For computing $D$-optimal stratified designs, there exists a rapid re-normalization heuristic, and extensive numerical computations suggest that it always converges to the optimum (see \cite{Harman14}). However, to our best knowledge, there is no analogous re-normalization heuristic for the general size-and-cost constrained problem \eqref{eq:Dopteq}. For instance, an obvious suggestion would be using an alternate application of the standard multiplicative algorithm (which could transform a design from $\bbQ^n_{++}$ outside of $\bbQ^n_{++}$), and the re-normalization described in Section \ref{sec:Examples} (which transforms any positive design back to $\bbQ^n_{++}$). Numerical experiments suggest that this method does not produce a convergent sequence of designs.

It is likely that the numerically most efficient method for solving  \eqref{eq:Dopteq} would combine the ideas of several methods. A simple practical approach is to use the barycentric algorithm with the deletion method in the initial part of the computation, which can significantly reduce the size of the design space, and then apply the SDP or the SOCP methods. Alternatively, one could try to combine the barycentric and vertex direction methods, similarly to \cite{Yu11}.

\subsection{Other criteria than $D$-optimality} Most considerations in the introduction apply also to other criteria than $D$-optimality, for instance to $A$-optimality. However, a barycentric algorithm for $A$-optimality has not yet been studied. It is probable that such an algorithm could be developed using methods analogous to \cite{Harman14} and that a generalization of the recent deletion method \cite{Pronzato13} could be used for the removal of the redundant design points. 

\section{Appendix}\label{sec:Appendix}

\begin{proof}[Proof of Proposition \ref{lem:baryweights}]
Let $x_+ \in \X_+$, $x_- \in \X_-$, $x_0 \in \X_0$ be fixed. Obviously, functions $\tw_{x_+x_-}$ and $\tw_{x_0}$ are non-negative on $\bbQ^n_+$ and positive on $\bbQ^n_{++}$ (note that $S(\bw)>0$ for all $\bw \in \bbQ^n_{++}$). The continuity of $\tw_{x_0}$ on $\bbQ^n_+$ is trivial. We will prove the continuity of $\tw_{x_+x_-}$ on $\bbQ^n_+$. 

For any $\bw \in \bbQ^n_+$ we have $w_{x_+} \leq \d^{-1}_{x_+}S(\bw)$, $w_{x_-} \leq \d^{-1}_{x_-}S(\bw)$ and \eqref{eq:baryweights_pm}, \eqref{eq:baryweights_0} yield the upper bound $\tw_{x_+ x_-}(\bw) \leq (\d^{-1}_{x_+}+\d^{-1}_{x_-})S(\bw)$. The only point of discontinuity of $\tw_{x_+ x_-}(\bw)$ could be $\bw^a \in \bbQ^n_+$ such that $S(\bw^a)=0$. But if some sequence $\{\bw^{(t)}\}_{t=0}^\infty$ of designs from $\bbQ^n_+$ converges to $\bw^a$, then, due to the continuity of $S(\cdot)$ on $\bbQ^n_+$, the upper bounds $(\d^{-1}_{x_+}+\d^{-1}_{x_-})S(\bw^{(t)})$ on $\tw_{x_+ x_-}(\bw^{(t)})$ converges to $0$. Consequently, applying the squeeze theorem, the non-negative numbers $\tw_{x_+ x_-}(\bw^{(t)})$ converge to $0=\tw_{x_+ x_-}(\bw^a)$. 

The normalization property \eqref{eq:bary1} of $\tw_{x_+x_-}$, $x_+ \in \X_+, x_- \in \X_-$, and $\tw_{x_0}$, $x_0 \in \X_0$ is straightforward to verify. We will check \eqref{eq:baryw}. Let $\bw \in \bbQ^n_+$ be such that $S(\bw)>0$. The component of the left-hand side of \eqref{eq:baryw} corresponding to $y_+ \in \X_+$ is
 \begin{equation*}
  \sum_{x_+ \in \X_+}\sum_{x_- \in \X_-}\tw_{x_+ x_-}(\bw)q^{(x_+,x_-)}_{y_+}+\sum_{x_0 \in \X_0} \tw_{x_0}(\bw)q^{(x_0)}_{y_+} = \sum_{x_- \in \X_-}  \frac{w_{y_+} w_{x_-}}{S(\bw)} \d_{x_-} = w_{y_+}.
\end{equation*}
 If $S(\bw)=0$, then $w_{x_+}=0$ for all $x_+ \in \X_+$ and $w_{x_-}=0$ for all $x_- \in \X_-$, which means that $\tw_{x_+ x_-}(\bw)=0$ for all $x_+,x_-$. Moreover, $q^{(x_0)}_{y_+}=0$ for all $x_0 \in \X_0$, that is, the left-hand side of \eqref{eq:baryw} is equal to zero, as required. Analogous proof is possible for $y_- \in \X_-$ and for $y_0 \in \X_0$.
\end{proof}

\begin{proof}[Proof of Theorem \ref{thm:convergence}]
To shorten the notation of some formulas, we will use $g_x(\bw):=d_x(\bw)-m$ for all $x \in \X$.

  Since $\bbQ^n_+$ is compact, the sequence $\{\bw^{(t)}\}_{t=0}^\infty \subset \bbQ^n_+$ has a limit point $\bw^{(\infty)} \in \bbQ^n_+$. Lemma 2 in \cite{Harman14} implies that non-singular matrices $\bM(\bw^{(t)})$ converge to some non-singular matrix $\bM^{(\infty)}$. From the continuity of $\bM(\cdot)$ it follows that $\bM^{(\infty)}=\bM(\bw^{(\infty)})$. Thus, $\bw^{(\infty)} \in \bbQ^n_r$ and the continuity of $\bM \to \bM^{-1}$ on the set of all non-singular $m \times m$ matrices gives: 
\begin{eqnarray}
  \lim_{t \to \infty}g_x(\bw^{(t)})&=&\lim_{t \to \infty}\bf^\top(x) \bM^{-1}(\bw^{(t)})\bf(x)-m \nonumber \\
  &=&\bf^\top(x) \bM^{-1}(\bw^{(\infty)})\bf(x)-m=g_x(\bw^{(\infty)}) \label{eq:convofg}
\end{eqnarray}  
for all $x \in \X$. Let
  \begin{equation*}
    x_+^* \in \mathrm{argmax}_{x_+ \in \X_+} \frac{g_{x_+}(\bw^{(\infty)})}{\d_{x_+}}, \: \:
    x_-^* \in \mathrm{argmax}_{x_- \in \X_-} \frac{g_{x_-}(\bw^{(\infty)})}{\d_{x_-}}, \: \: x_0^* \in \mathrm{argmax}_{x_0 \in \X_0} d_{x_0}(\bw^{(\infty)}).
  \end{equation*}

Note that there exists a constant $\gamma>0$ such that for any $t \in \bbN$ and any $x_+ \in \X_+,x_- \in \X_-$:
  \begin{equation*}
    \tilde{d}_{x_+ x_-}(\bw^{(t)}) =
    \frac{\frac{g_{x_+}(\bw^{(t)})}{\d_{x_+}}+\frac{g_{x_-}(\bw^{(t)})}{\d_{x_-}}}{\frac{1}{\d_{x_+}}+\frac{1}{\d_{x_-}}}+m \geq \gamma \left( \frac{g_{x_+}(\bw^{(t)})}{\d_{x_+}}+\frac{g_{x_-}(\bw^{(t)})}{\d_{x_-}} \right) + m,
  \end{equation*}
  which gives
  \begin{equation}\label{eq:d1}
  d^{\pi}_{x^*_+}(\bw^{(t)}) \geq \frac{\gamma}{m} \left( \frac{g_{x^*_+}(\bw^{(t)})}{\d_{x^*_+}} + \frac{\sum_{x_- \in \X_-} w^{(t)}_{x_-}g_{x_-}(\bw^{(t)})}{S(\bw^{(t)})} \right) + 1
  \end{equation}
  and similarly
  \begin{equation}\label{eq:d2}
  d^{\pi}_{x^*_-}(\bw^{(t)}) \geq \frac{\gamma}{m} \left( \frac{g^*_{x_-}(\bw^{(t)})}{\d_{x^*_-}} + \frac{\sum_{x_+ \in \X_+} w^{(t)}_{x_+}g_{x_+}(\bw^{(t)})}{S(\bw^{(t)})} \right) + 1. 
  \end{equation}
 
  Assume that $\bw^{(\infty)}$ is not $D$-optimal in $\bbQ^n_+$. Then, using Theorem \ref{thm:equivalence}, we see that
  \begin{eqnarray*}
   && \text{either (i) } g_{x^*_0}(\bw^{(\infty)}) > 0,\\
   && \text{or (ii) } g_{x^*_0}(\bw^{(\infty)}) \leq 0,  \text{ and } \frac{g_{x_+^*}(\bw^{(\infty)})}{\d_{x_+^*}}+\frac{g_{x_-^*}(\bw^{(\infty)})}{\d_{x_-^*}}>0.  
  \end{eqnarray*}

 Assume (i). From \eqref{eq:convofg} with $x=x^*_0$ we see that there exist some $t_1 \in \bbN$ and $\epsilon>0$ such that $g_{x^*_0}(\bw^{(t)}) \geq \epsilon m$ for all $t \geq t_1$, i.e., $m^{-1}d_{x^*_0}(\bw^{(t)}) \geq 1+\epsilon$ for all $t \geq t_1$. But the transformation rules \eqref{eq:bary} and \eqref{eq:baryzero} of the barycentric algorithm give
\begin{equation*} 
 w^{(t)}_{x^*_0}=\prod_{s=t_1}^{t-1} m^{-1}d_{x^*_0}(\bw^{(s)})w_{x^*_0}^{(t_1)} \geq (1+\epsilon)^{t-t_1}w^{(t_1)}_{x^*_0},
\end{equation*}
which converges to infinity for $t \to \infty$. This is a contradiction because $w^{(t)}_{x^*_0} \leq 1$ for all $t \in \bbN$.
  
  Assume (ii). From \eqref{eq:convofg} with $x=x^*_+$ and $x=x^*_-$ we see that there is some $t_1 \in \bbN$ and $\epsilon>0$ such that for all $t \geq t_1$:
  \begin{equation}\label{eq:eps}
    \frac{g_{x_+^*}(\bw^{(t)})}{\d_{x_+^*}}+\frac{g_{x_-^*}(\bw^{(t)})}{\d_{x_-^*}} \geq \epsilon.
  \end{equation}
  For all $t \in \bbN$ the simple equality $\sum_{x \in \X} w_xd_x(\bw^{(t)})=m$ yields
  \begin{equation} \label{eq:zerosum}
   \frac{\sum_{x_+ \in \X_+}w^{(t)}_{x_+}g_{x_+}(\bw^{(t)})}{S(\bw^{(t)})} + \frac{\sum_{x_- \in \X_-} w^{(t)}_{x_-}g_{x_-}(\bw^{(t)})}{S(\bw^{(t)})} = -
 \frac{ \sum_{x_0 \in \X_0} w^{(t)}_{x_0}g_{x_0}(\bw^{(t)})}{S(\bw^{(t)})}.
  \end{equation}
  From the assumption $g_{x^*_0}(\bw^{(\infty)}) \leq 0$ we have $g_{x_0}(\bw^{(\infty)}) \leq 0$ for all $x_0 \in \X_0$, i.e., \eqref{eq:convofg} implies that for all $x_0 \in \X_0$ the sequence $\{g_{x_0}(\bw^{(t)})\}_{t=0}^{\infty}$ converges to some non-positive value. Because the weights $w^{(t)}_{x_0}$ are bounded and we assume that $\liminf_{t \to \infty}S(\bw^{(t)})>0$, it is clear that the limit inferior of the right-hand side of \eqref{eq:zerosum} is non-negative. Therefore, \eqref{eq:zerosum} ensures that there is $t_2 \in \bbN$ such that for all $t \geq t_2$:
    \begin{equation}\label{eq:halfeps}
  \frac{\sum_{x_+ \in \X_+} w^{(t)}_{x_+}g_{x_+}(\bw^{(t)})}{S(\bw^{(t)})}+\frac{\sum_{x_- \in \X_-} w^{(t)}_{x_-}g_{x_-}(\bw^{(t)})}{S(\bw^{(t)})} \geq -\frac{\epsilon}{2}.
  \end{equation}
   
  Summing \eqref{eq:d1} with \eqref{eq:d2}, then, using \eqref{eq:eps} and \eqref{eq:halfeps}, we see that for all $t \geq \max(t_1,t_2)$:
  \begin{equation}\label{eq:boundonsum}
    d_{x_+^*}^{\pi}(\bw^{(t)})+d_{x_-^*}^{\pi}(\bw^{(t)}) \geq 2+\frac{\gamma \epsilon}{m}-\frac{\gamma (\epsilon/2)}{m}=2+\frac{\gamma \epsilon}{2m}.
  \end{equation}
  At the same time, using \eqref{eq:d1} and \eqref{eq:zerosum}, we obtain
  \begin{eqnarray}
      d_{x_+^*}^{\pi}(\bw^{(t)}) &\geq & \frac{\gamma}{m} \left(\frac{g_{x_+^*}(\bw^{(t)})}{\d_{x_+^*}} - \frac{\sum_{x_+ \in \X_+} w^{(t)}_{x_+}\d_{x_+}\frac{g_{x_+}(\bw^{(t)})}{\d_{x_+}}}{S(\bw^{(t)})} - 
     \frac{\sum_{x_0 \in \X_0} w^{(t)}_{x_0}g_{x_0}(\bw^{(t)})}{S(\bw^{(t)})} \right) + 1 \nonumber\\
      &\geq& \frac{\gamma}{m} \left( \frac{g_{x_+^*}(\bw^{(t)})}{\d_{x_+^*}} - \max_{x_+ \in \X_+}\frac{g_{x_+}(\bw^{(t)})}{\d_{x_+}}- \frac{\sum_{x_0 \in \X_0} w_{x_0}^{(t)}g_{x_0}(\bw^{(t)})}{S(\bw^{(t)})} \right) + 1. \label{eq:dpiineq}
  \end{eqnarray}  
 We again obtained the term that appeared at the right-hand side of \eqref{eq:zerosum}, and, as we have already shown, its limit inferior is non-negative. Note also that from \eqref{eq:convofg} and from the definition of $x^*_+$ we have
 \begin{equation*}
  \lim_{t \to \infty} \left( \frac{g_{x_+^*}(\bw^{(t)})}{\d_{x_+^*}} - \max_{x_+ \in \X_+}\frac{g_{x_+}(\bw^{(t)})}{\d_{x_+}} \right)=0.
 \end{equation*}
 Thus, \eqref{eq:dpiineq} proves that the limit inferior of $d_{x_+^*}^{\pi}(\bw^{(t)})$ is greater or equal to $1$. Similarly, it can be shown that $d_{x_-^*}^{\pi}(\bw^{(t)})$ has also limit inferior greater or equal to $1$. Therefore, we have 
 \begin{equation*}
   \min (d_{x_+^*}^{\pi}(\bw^{(t)}),d_{x_-^*}^{\pi}(\bw^{(t)})) \geq  \left(1+ \frac{\gamma\epsilon}{8m}\right)/\left(1+ \frac{\gamma\epsilon}{4m}\right)
  \end{equation*}
 for all sufficiently large $t$. Using this fact together with \eqref{eq:boundonsum} we see that there exists some $t_3 \in \bbN$, such that
  \begin{equation*}
   d_{x_+^*}^{\pi}(\bw^{(t)})d_{x_-^*}^{\pi}(\bw^{(t)}) \geq \frac{d_{x_+^*}^{\pi}(\bw^{(t)}) + d_{x_-^*}^{\pi}(\bw^{(t)})}{2} \min (d_{x_+^*}^{\pi}(\bw^{(t)}),d_{x_-^*}^{\pi}(\bw^{(t)})) \geq 1+ \frac{\gamma\epsilon}{8m}
  \end{equation*} 
   for all $t > t_3$. Hence, the definition of $\bd^{\pi}(\bw^{(t)})$ and the form of the transformations \eqref{eq:baryplus}, \eqref{eq:baryminus} imply that
\begin{eqnarray*}  
  w^{(t)}_{x_+^*}w^{(t)}_{x_-^*}=\left(\prod_{s=t_3}^{t-1} d_{x_+^*}^{\pi}(\bw^{(s)})d_{x_-^*}^{\pi}(\bw^{(s)}) \right) w^{(t_3)}_{x_+^*}w^{(t_3)}_{x_-^*} \geq \left(1+ \frac{\gamma\epsilon}{8m}\right)^{t-t_3} w^{(t_3)}_{x_+^*}w^{(t_3)}_{x_-^*},
\end{eqnarray*}  
 which converges to infinity as $t \to \infty$. This is a contradiction since $w^{(t)}_{x_+^*}w^{(t)}_{x_-^*} \leq 1$ for all $t \in \bbN$.
  
  Consequently, the limit point $\bw^{(\infty)}$ of $\{\bw^{(t)}\}_{t=0}^\infty$ is $D$-optimal in $\bbQ^n_+$, which implies the statement of Theorem \ref{thm:convergence}.
\end{proof}

\begin{proof}[Proof of Lemma \ref{lem:guarantee}]
Let $\liminf_{t \to \infty} S(\bw^{(t)})=0$. Compactness of $\bbQ^n_+$ guarantees that there exists some increasing sequence $\{t_i\}_{i=1}^{\infty}$ of natural numbers, such that $\lim_{i \to \infty} S(\bw^{(t_i)})=0$ and, simultaneously, $\lim_{i \to \infty}\bw^{(t_i)}=\bw^{(\infty)}$ for some $\bw^{(\infty)} \in \bbQ^n_+$. However, $\lim_{i \to \infty} S(\bw^{(t_i)})=0$ implies $\lim_{i \to \infty} w^{(t_i)}_{x_+}=0$ for all $x_+ \in \X_+$ and $\lim_{i \to \infty} w^{(t_i)}_{x_-}=0$ for all $x_- \in \X_-$, i.e., $\bw^{(\infty)}$ has all components zero, except for some $x_0 \in \X_0$, which means that $\phi(\bw^{(\infty)}) \leq v_0$. Thus, the continuity of $\phi$ yields $\lim_{i \to \infty} \phi(\bw^{(t_i)}) = \phi(\bw^{(\infty)}) \leq v_0$. But the sequence $\{\phi(\bw^{(t)})\}_{t=0}^\infty$ of criterial values of designs generated by the barycentric algorithm is non-decreasing, therefore $\phi(\bw^{(s)}) \leq v_0$ for all $s \in \{0,1,2,...\}$. This contradicts an assumption of the lemma, namely $\phi(\bw^{(s)}) > v_0$ for some $s \in \{0,1,2,...\}$.
\end{proof}

\section*{Acknowledgements}
This research was supported by the Slovak VEGA-Grant No. 1/0163/13. 


\end{document}